\documentclass[oneside, reqno, 10pt]{amsart}

\usepackage{amsmath,amssymb}
\usepackage[dvips]{geometry}
\usepackage{color}
\numberwithin{equation}{section}
\theoremstyle{plain}                
\newtheorem{theorem}{Theorem}[section]
\newtheorem{lemma}[theorem]{Lemma}
\newtheorem{proposition}[theorem]{Proposition}
\newtheorem{corollary}[theorem]{Corollary}
\theoremstyle{definition}           
\newtheorem{definition}[theorem]{Definition}
\newtheorem{example}[theorem]{Example}
\newtheorem{assumption}[theorem]{Assumption}
\theoremstyle{remark}
\newtheorem{remark}[theorem]{Remark}

\newcommand{\revision}[1]{#1}
\newcommand{\nothing}[1]{}

\newcommand{\tot}{\tfrac{1}{2}} 
\newcommand{\abs}[1]{\left| #1 \right|} 
\newcommand{\set}[1]{\left\{#1\right\}} 
\newcommand{\sets}[2]{\set{#1\,:\,#2}} 
\newcommand{\inds}[1]{ {\mathbf 1}_{\set{#1}}} 
\newcommand{\seq}[1]{\set{#1_n}_{n\in\N}} 

\newcommand{\norm}[1]{{||#1||}} 

\newcommand{\ft}[2]{#1\dots#2} 
\renewcommand{\ft}[2]{#1,\dots,#2}

\providecommand{\R}{} \renewcommand{\R}{{\mathbb R}}

\newcommand{\N}{{\mathbb N}}
\newcommand{\PP}{{\mathbb P}}
\newcommand{\QQ}{{\mathbb Q}}

\newcommand{\EE}{{\mathbb E}}
\newcommand{\FF}{{\mathcal F}}

\renewcommand{\AA}{{\mathcal A}}
\newcommand{\bsy}[1]{\boldsymbol{#1}}
\newcommand{\bS}{\boldsymbol{\Sigma}}
\newcommand{\bmu}{\boldsymbol{\mu}}

\newcommand{\ld}{\lambda}

\newcommand{\pds}[1]{\frac{\partial}{\partial #1}}

\newcommand{\bone}{\bsy{1}}
\newcommand{\bsi}{\bsy{\sigma}}
\newcommand{\balpha}{\bsy{\alpha}}
\newcommand{\bzeta}{\bsy{\zeta}}
\newcommand{\bze}{\bzeta}
\newcommand{\bzt}{\bzeta^T}
\newcommand{\brho}{\bsy{\rho}}

\newcommand{\bpi}{\bsy{\pi}}

\newcommand{\Xze}{X^{\bzeta}}
\newcommand{\bzs}{\bzeta^*}
\newcommand{\bzi}{\bzeta^{\infty}}

\newcommand{\Xzs}{X^{\bzs}}
\newcommand{\Xzi}{X^{\bzi}}
\newcommand{\Yze}{Y^{\bze}}

\DeclareMathOperator\Range{Range}

\newcommand{\optbz}{\bzeta^*}
\newcommand{\Xoze}{X^{\optbz}}
\renewcommand{\bS}{\bsy{S}}
\newcommand{\bW}{\bsy{W}}

\newcommand{\cprfi}[1]{\{ #1 \}_{t\in [0,\infty)}}

\newcommand{\bt}{\bar{\tau}}

\newcommand{\vnot}[2]{ (#1)_{i=\ft{1}{#2}} }
\newcommand{\vnotr}[2]{ (#1)_{j=\ft{1}{#2}} }

\newcommand{\mnot}[3]{(#1)_{i=\ft{1}{#2}}^{j=\ft{1}{#3}}}

\newcommand{\zmu}{\zeta_{\bmu}}
\newcommand{\zsigma}{\zeta_{\bsi}}

\newcommand{\var}{\mathrm{VaR}}
\newcommand{\tvar}{\mathrm{TVaR}}
\newcommand{\lel}{\mathrm{LEL}}

\newcommand{\fa}{F^{\mathrm{abs}}}
\newcommand{\fr}{F^{\mathrm{rel}}}
\newcommand{\tfa}{\tilde{F}^{\mathrm{abs}}}
\newcommand{\tfr}{\tilde{F}^{\mathrm{rel}}}
\newcommand{\efa}{f^{\mathrm{abs}}}
\newcommand{\eha}{h^{\mathrm{abs}}}

\renewcommand{\aa}{a^{\mathrm{abs}}}
\newcommand{\ar}{a^{\mathrm{rel}}}

\newcommand{\Az}{\mathrm{Tr}_{X}^{\bze}}
\newcommand{\Azx}{\mathrm{Tr}_{X}}

\newcommand{\Amz}{\mathrm{Tr}_{A}^{\bze}}
\newcommand{\tbmu}{\brho_M(t)}
\newcommand{\tbrho}{\bar{\brho}}
\newcommand{\hro}{\hat{\brho}} 
\newcommand{\bs}{\beta^*}
\newcommand{\bi}{\beta^{\infty}}
\newcommand{\bxt}{\beta(t,X(t))}

\newcommand{\bzr}{\bzeta^r}
\newcommand{\Xzr}{X^{\bzr}}
\newcommand{\Mze}{M^{\bze}}
\newcommand{\vp}{\varphi}

\begin{document}

\begin{center}
  \huge\sc Maximizing the Growth Rate under Risk Constraints\footnote{
    The authors would like to thank Steve Shreve for helpful ideas,
    engaging conversations and valuable guidance in the course of
    writing this paper.

    This material is based upon work supported by the National Science
    Foundation under Grant Numbers 0103814, 0139911, and 0404682.  Any
    opinions, findings, and conclusions or recommendations expressed
    in this material are those of the author(s) and do not necessarily
    reflect the views of the National Science Foundation.}
\end{center}

\medskip

\begin{center}
\today
\end{center}

\medskip

\begin{minipage}[c]{0.45\textwidth}
\begin{center}
{\bf Traian A.~Pirvu}\\ 
Mathematics Department \\
The University of British Columbia\\ 
Vancouver, BC, Canada\\
{\tt tpirvu@math.ubc.ca}
\end{center}
\end{minipage}
\begin{minipage}{0.4\textwidth}
\begin{center}
{\bf Gordan \v Zitkovi\' c}\\
Department of Mathematics\\ 
University of Texas at Austin\\ 
Austin, TX, USA\\
{\tt gordanz@math.utexas.edu}
\end{center}
\end{minipage}
\ \\[3ex]
\begin{quote} {\bf Abstract:} { We investigate the ergodic problem of
    growth-rate maximization under a class of risk constraints in the
    context of incomplete, It\^ o-process models of financial markets
    \revision{with random ergodic coefficients}.  Including {\em value-at-risk}
    (VaR), {\em tail-value-at-risk} (TVaR), and {\em limited expected
      loss} (LEL), these constraints can be both wealth-dependent
    (relative) and wealth-independent (absolute). The optimal policy
    is shown to exist in an appropriate admissibility class, and can
    be obtained explicitly by uniform, state-dependent scaling down of
    the unconstrained (Merton) optimal portfolio. This implies that
    the risk-constrained wealth-growth optimizer locally behaves like
    a CRRA-investor, with the relative risk-aversion coefficient
    depending on the current values of the market coefficients.  }
\end{quote}\ \\[2ex]
\\[3ex]
\indent {\bf 2002 AMS Classification:} 91B30, 60H30, 60G44. \\[3ex]

\pagestyle{myheadings} \thispagestyle{plain} \markboth{ T.~PIRVU AND
  G.~\v ZITKOVI\' C}{MAXIMIZING THE GROWTH RATE UNDER RISK
  CONSTRAINTS}

\newpage 
\section{Introduction}

The problem of dynamic portfolio choice has received a great deal of
attention since the seminal work of Merton \cite{Mer69,Mer71}. It has
caught the attention of both the mathematical and the financial
research community because of its interesting technical aspects as well
as its practical applicability. The present paper aims to contribute
to both of these features by considering a constrained ergodic-control
problem, where the constraint - taken directly from the everyday
financial practice - exhibits an unexpected degree of structure in its
interplay with the objective function.  Specifically, our aim is to
maximize the long-term growth rate
\[\liminf_{t\to\infty} \tfrac{1}{t} \log(X_t)\] of the investor's 
wealth $X_t$ under several {\em risk constraints}: regulatory
agencies, as well as the internal institutional policies, often
require the risk inherent in the trading strategies of investors to be
carefully monitored and kept under control. Among a myriad of risk
measures employed by both academics and practitioners, the most
recognized one is without a doubt the Value-at-Risk (VaR) which
measures the magnitude of a percentile of the loss distribution. Due
to several shortcomings (e.g., lack of convexity and insensitivity to
catastrophic losses), the use of VaR has recently been complemented by
other measures of risk ({\em Tail-VaR} (TVaR), for example). This
development prompted us to try to consider a large class of measures
of risk (containing both Var and TVaR) and study the growth-rate
maximization problem where the risk in the wealth is constrained by
one of these measures.

\subsection*{Existing research}
The study of optimal control problems where the growth rate
$\tfrac{1}{t} \log(X_t)$ of a certain controlled quantity $X_t$ is to
be maximized goes back at least to Kelly \cite{Kel56} and Breiman
\cite{Bre61}. While such problems have been studied in a variety of
settings, we focus here on the applications in finance and economics.
The earliest discrete-time results were established by Hakansson
\cite{Hak70} and Thorp \cite{Tho71}, while Karatzas \cite{Kar89}
studied the continuous-time version \revision{where the 
stocks follow an It\^{o} process}. Aase and {\O}ksendal \cite{AasOks88}
extend the existing results to allow stock prices to jump. Taksar,
Klass, and Assaf \cite{TalKlaAss88}, Pliska and Selby \cite{PliSel94}
and Akian, Sulem, and Taksar \cite{AkiSulTak01} address this problem
in the presence of transaction costs \revision{in a Black-Scholes
  model.} 
Algoet and Cover
\cite{AlgCov88} and Cover \cite{Cov84,Cov91} provide  algorithms for
maximizing the growth rate of a portfolio in a very general
discrete-time model. Jamshidian \cite{Jam91} examines the behavior of
this algorithm in a continuous-time \revision{diffusion} model. 
In a sequence of papers,
Fleming and Sheu \cite{FleShe99, FleShe00, FleShe04} reformulate this
problem as an infinite time horizon risk-sensitive control problem \revision{in a diffusion paradigm.}

The problem of maximizing the growth rate of a portfolio is an example
of an ergodic stochastic control problem. The study of ergodic control
dates back to Bellman \cite{Bel57}, who considered the discrete-time
case. For the continuous-time theory, see, for example, Lasry
\cite{Las74}, Tares \cite{Tar82,Tar85}, and Cox and Karatzas
\cite{CoxKar85}.  An interesting reference is a survey paper by Robin
\cite{Rob83}.

From a decision-theoretic point of view, the maximization of the
growth rate (which is essentially equivalent to the maximization of
the expected logarithmic utility) at one point seemed as a natural
choice of the objective function for money managers.  However, it did
not take long for the community to turn a critical eye towards the
high degree of risk inherent in such strategies.  To make his point
perfectly clear, Samuelson \cite{Sam79} argues in words of {\em
  literally one syllable} that maximization of non-logarithmic
utilities with finite time horizons should be adopted as a more
desirable goal. He says:
\begin{quote}
  \em He who acts in $N$ plays to make his mean log of wealth as big
  as it can be made will, with odds that go to one as $N$ soars, beat
  me who acts to meet my own tastes of risk.
\end{quote}
but then he adds:
\begin{quote}
  \em When you lose - and sure can lose - with $N$ large, you can lose
  real big.
\end{quote}

It is, therefore, only natural that a new line of research - one
attempting to subdue the excessive risk from growth-maximization - has
soon emerged.  Grossman and Zhou \cite{GroZho93} and Cvitani\' c and
Karatzas \cite{CviKar95} study this optimization problem under
so-called ``drawdown constraints'', where the wealth process is never
allowed to fall below a fixed fraction of its maximum-to-date,
\revision{and the risky assets follow an It\^{o} process}.  Closer to
our setting is the work of MacLean, Sanegre, Zhao and Ziemba
\cite{MacSanZhaZie04} who consider a discrete time set-up, where the
maximization of capital growth subject to Value at Risk constraint is
studied by means of multistage stochastic programming. The literature
on the continuous-time models with risk constraints of VaR-type is
much broader.  Basak and Shapiro \cite{BasSha01} analyze the optimal
dynamic portfolio and wealth-consumption policies of utility
maximizing investors who use VaR to manage their risk exposure,
\revision{in a complete-market It\^ o-process framework.  Arguing
  informally, they guess the solution and discuss it without providing
  an existence proof.} One of their \revision{(heuristic)} findings is
that VaR-constrained risk managers actually increase exposure to risky
assets compared to the unconstrained case and, to borrow a phrase from
the abstract of \cite{BasSha01}, {\em ``consequently incur larger
  losses when losses occur''}.  In order to fix this deficiency, they
choose another risk measure based on the risk-neutral expectation of a
loss - the \textit{Limited Expected Loss} (LEL).  A drawback of their
model is that the VaR is computed in a static manner, and never
reevaluated after the initial date.  Emmer, Kl\"{u}ppelberg and Korn
\cite{EmmKluKor01} consider a dynamic model with
\textit{Capital-at-Risk} (a version of VaR) limits, \revision{in the
  Black-Scholes-Samuelson model.}  However, the assumption that
portfolio proportions are held fixed during the whole investment
period leads to a similar problem.  Dmitrasinovi\' c-Vidovi\' c,
Lari-Lavassani, Li and Ware \cite{DmiLarLiWar03} extend
\cite{EmmKluKor01} to the case of time dependent,
\revision{deterministic,} parameters and investment strategies, where
analytical formulas for the optimal strategies are
obtained. \revision{
 Gabih,
\revision{Grecksch} and Wunderlich \cite{GabGreWun05} follow
\cite{BasSha01} and \revision{extend} their results to cover the case
of a bounded expected loss and provide detailed solutions for the
class of CRRA utilities, \revision{in a constant coefficients market
  model. They employ the martingale method to establish the optimal
  portfolios under constraints and conclude with some numerical
  results.} Gundel and Weber \cite{GunWeb06} analyze optimal portfolio
choice of utility maximizing agents in a general continuous-time
financial market model under a joint budget and the downside risk
constraint measured by an abstract convex risk measure (the VaR
constraint used in \cite{BasSha01} is a particular type of a downside
risk constraint and it can be reformulated in the language of
translation invariant risk measures). \revision{ The utility
  maximization problem under these constraints is solved in closed
  form, and the conditions under which the 
the constraints are binding are determined.}
}

\revision{An axiomatic approach to risk-measurement has started with
  the seminal paper of Artzner, Delbaen, Eber and Heath
  \cite{ArtDelEbeHea99}, where four simple postulates (to be satisfied
  by any risk measures) are proposed. The resulting functionals are
  termed {\em coherent risk measures}. F\"{o}llmer and Schied
  \cite{FolSch04} relax one of the axioms of \cite{ArtDelEbeHea99} and
  obtain a more general notion of a {\em convex risk
    measures}. Jaschke and K\"uchler investigate the properties of
  convex risk measures in \cite{JasKuc01}.  All of the work mentioned
  above assumed static modeling framework. The literature on dynamic
  risk measures (where a temporal component is added) is relatively
  new and we only mention a small portion of the existing research: in
  \cite{ArtDelEbeHeaKu02} and \cite{ArtDelEbeHeaKu04}, Artzner,
  Delbaen, Eber, Heath and Ku construct coherent risk measures on
  stochastic processes rather than on random variables. Wang
  \cite{Wan03} considers a set of axioms for dynamic risk measures and
  analyzes the class of measures satisfying his axioms. Delbaen,
  Cheridito and Kupper \cite{DelCheKup04} investigate the properties
  of risk measures defined over stochastic processes.}  Another
approach to modeling of risk constraints - also developed with the
intention of going beyond the static formulation and building on the
work of \cite{BasSha01} - was introduced by Cuoco, He and Issaenko
\cite{CuoHeIss07}. A more realistic, dynamically-consistent model of
optimal behavior of a trader subject to risk constraints is presented
here: the authors assume that the risk in the trading portfolio is
reevaluated dynamically, using the current information. Hence, the
trader must continuously monitor his/her trading strategy in order to
honor the risk limits at every instant.  Another assumption made in
\cite{CuoHeIss07} is that when assessing the risk of a portfolio, the
distribution of the portfolio composition is kept unchanged over a
given horizon $\tau$ (more precisely, the relative exposures to
different assets are kept unchanged). In other words, the model
outlaws those trading strategies which at any point $t$ in time, if
kept constant over the time interval $[t,t+\tau]$, would result in a
loss whose VaR is below a given threshold.  The authors perform an
analogous analysis with VaR replaced by TVaR, and establish that it is
possible to identify a dynamic VaR risk limit equivalent to a given
TVaR risk limit. Finally, they conclude that that the risk exposure of
a trader subject to VaR or TVaR risk limits is always lower than that
of an unconstrained trader. \revision{We should note that, while this
  paper makes a very significant contribution to modeling of risk
  constraints, it limits its scope to a multivariate Black-Scholes
  markets. Relaxing this condition is one of the main motivations for
  the research that lead to the present paper. }

Cuoco and Liu \cite{CuoLiu03} study the dynamic investment and
reporting problem of a financial institution subject to capital
requirements based on self-reported VaR estimates. For a market with
constant price coefficients, they show that optimal portfolios display
a local three-fund property. Leippold, Trojani and Vanini
\cite{LeiVanTro02} analyze VaR-based regulation rules and their
possible distortion effects on financial markets \revision{in the
  setting of diffusion processes}.
They show that in
partial equilibrium the effectiveness of VaR regulation is closely
linked to the \textit{leverage effect} - the tendency of volatility to
increase when the prices decline. Berkelaar, Cumperayot and Kouwenberg
\cite{BerCumKou02} study the effect of VaR-based risk management on
asset prices, (\revision{modelled as It\^{o} processes}) and the volatility smile. They look at an equilibrium
model where a portion of the agents are constrained with VaR. It turns
out that in equilibrium VaR reduces market volatility, but in some
cases raises the probability of extreme losses.

In \cite{Yiu04}, the author considers an optimal investment problem, where an
agent maximizes utility of his/her intertemporal consumption over a
period of time under a dynamic VaR constraint. A numerical method is
proposed to solve the corresponding HJB-equation. He finds that, under
the optimal strategy, the investment in risky assets is reduced by the
VaR constraint. Atkinson and Papakokinou \cite{AtkPap05} derive the
solution to the optimal portfolio and consumption problem subject to
CaR ({\em Capital-at-Risk}) and VaR constraints by using stochastic
dynamic programming. In both \cite{Yiu04} and \cite{AtkPap05} the
strong assumption of constant market coefficients is imposed.

\subsection*{Our contributions}
In this work we follow the approach of \cite{CuoHeIss07} and impose
dynamic risk constraints of the VaR-type. Unlike \cite{CuoHeIss07}, we
maximize the long-term (ergodic) growth rate of the accumulated
wealth. Moreover, our model allows market coefficients (the stock
price return and volatility) to be random processes, assumed to
satisfy a mild ergodicity condition, but without any restriction on
the completeness of the resulting market. In addition to the
constant-coefficient models, our set-up allows for a wide range of
stochastic-volatility and seasonally-varying models.

Consequently, the risk measurement on the time interval $[t, t+\tau]$
is performed under the assumption that the market coefficients, as
well as portfolio proportions, are held constant at their value at
time $t$.  While, for the sake of simplicity, our risk constraint is
taken to be either VaR, TVaR or LEL, all our results hold under a more
general class of risk measures, expressible as deterministic functions
of two ``sufficient statistics'': \textit{portfolio return} and
\textit{portfolio volatility}.  Furthermore, we differentiate between
two different risk-limit implementations - relative and absolute
(depending on whether the risk is measured as a percentage of current
wealth, or in dollar terms). In the latter case, the constraints
become wealth- (state-) dependent and the agent finds him-/herself in
an interesting predicament - should he/she maximize the current growth
rate of wealth, or act more conservatively and thus face more
favorable constraints in the future. This raises the complexity level
of the problem considerably and requires a delicate mathematical
analysis, the final conclusion of which is that nothing can be gained
by waiting. More precisely, the structure of the aforementioned
wealth-dependent case is such that the constraints are not binding, as
long as the wealth is below a certain level.  Once the wealth gets
above that level, the constraints become binding and the set of
admissible portfolios reduces as wealth accumulates.  In the limit as
wealth approaches infinity, the constraints shrink and approach the
{\em limiting constraint set} (which still depends on the market
coefficients, time, and the current state of the world).  One of
optimal strategies we identify can be described as follows: pretend
that the limiting constraint is imposed from the start and simply
project (under a specific metric) the unconstrained optimal portfolio
(the {\em Merton proportion process}) onto it. Alternatively,
projecting the Merton proportion process onto the current constraint
set leads to the same ergodic behavior.

In the relative case, we show that the projection of the Merton
proportion onto the (current) constraint set describes the optimal
behavior in both ergodic, and the finite-horizon cases.

Thanks to the special structure of the constraints, the projection of
the Merton proportion onto the constraint is collinear with the origin
and Merton proportion itself.  This fact is the key to the success of
our analysis and sheds new light on the reasons why VaR constraints,
coupled with the growth-rate maximization, leads to such agreeable
results. Moreover, the ratio between the norm of the projection and
the norm of the Merton proportion can be interpreted as the reduction
in risk-exposure of the constrained agent (compared to the
unconstrained one). This number will follow a random process
$\beta_t$, thus making our agent act locally as if he/she is a
CRRA-utility maximizer with the coefficient of relative risk aversion
depending on the current market conditions. Interestingly, we show
that $\beta_t$ is a nonlinear deterministic function $\delta$ of the
norm of the Merton proportion process only. Furthermore, the value of
the optimal growth-rate of wealth can be obtained by integrating the
real function $x^2\delta(x)$ against the invariant measure of the norm
of the Merton proportion process.

\subsection*{Organization of the paper and some remarks on notation and terminology}
The reminder of this paper is organized as follow. In section $2$ we
describe the financial market model, the risk measures and the
constraint sets. Section $3$ contains the main results, and Section
$4$ develops the proof of the main theorem through a number of
auxiliary results.  The paper ends with an Appendix containing some
technical results.

All random processes and random fields in the paper possess the degree
of measurability sufficient for all the operations preformed on them.
We do not mention, or check, this fact in the main body, leaving the
standard proofs to the interested reader. Occasionally, a phrase like
``pick a typical $\omega\in\Omega$'' will be used. It will mean that all
the previous statements, proven to hold a.s., are assumed to 
hold for this particular realization $\omega\in\Omega$.

A stochastic processes $\cprfi{X_t}$ will usually be denoted simply by
$X_t$ (or even $X$), and the elements of $\R^n$ or $\R^m$ will be
interpreted as column vectors in the relevant contexts.

\section{ Model Description and Problem Formulation}

\subsection{ The Financial Market} Our model of a financial market,
based on a a filtered probability space
$(\Omega,\FF,\cprfi{\FF_t},\PP)$ satisfying the usual conditions,
consists of $n+1$ assets. The first one, $\cprfi{S_0(t)}$, is a {\em
  riskless bond} with a strictly positive constant interest rate $r>0$. The
remaining $n$ are referred to as {\em stocks}, and are modelled by an
$n$-dimensional It\^o-process
$\cprfi{\bS(t)}=\cprfi{\vnot{S_i(t)}{n}}$. The dynamics of their
evolution is determined by the following stochastic differential
equations in which $\cprfi{\bW(t)}=\cprfi{\vnot{W_i(t)}{m}}$ is an
$m$-dimensional standard Brownian motion:
\begin{equation}%
\label{equ:SDSs-for-stocks}
\left.    
\begin{aligned}
  dS_0(t)& = S_0(t) r \, dt \\
  dS_i(t)& =S_i(t)\Big( \alpha_i(t)\, dt+\sum_{j=1}^m \sigma_{ij}(t)\,
  dW_j(t) \Big),\ i=\ft{1}{n},
    \end{aligned}
\right\}, t\in [0,\infty),
\end{equation}
where $\cprfi{\balpha(t)}=\cprfi{\vnot{\alpha_i(t)}{n}}$ is an $\R^n$-valued
{\em mean rate of return } processes, and $\cprfi{\bsi(t)}$
$=\cprfi{\mnot{\sigma_{ij}(t)}{n}{m}}$ is an $n\times m$-matrix-valued
{\em variance-covariance} process. In order for the equations in
\eqref{equ:SDSs-for-stocks} to be well-defined, we impose the
following regularity conditions on the coefficient processes
$\balpha(t)$ and $\bsi(t)$:
\begin{assumption}
  \label{ass:basic-regularity}
 All the components of the 
 processes $\cprfi{\balpha(t)}$ and $\cprfi{\bsi(t)}$ are
 c\' agl\' ad (left-continuous with right limits). 
\end{assumption}

\begin{remark}\   
\begin{enumerate} 
\item
The c\' agl\' ad requirement from Assumption
\ref{ass:basic-regularity} is used in several different ways in this
paper. First, it ensures local boundedness, a property needed in
several parts of the proof of the main result. Second, it is necessary
for the standard SDE theory (see Lemma \ref{lem:SDE-exists} below) to be
applicable. Finally, it directly implies the following 
integrability condition
\begin{equation}
    \nonumber 
    \begin{split}
      \sum_{i=1}^n \int_0^t \abs{\alpha_i(u)}\, du+\sum_{i=1}^n
      \sum_{j=1}^m \int_0^t \sigma_{ij}(u)^2\, du<\infty,\text{ for
        all $t\in [0,\infty)$, a.s.}
    \end{split}
\end{equation}
\item Further distributional 
restrictions will be imposed on $\bsi(t)$ and $\balpha(t)$ in
the sequel (the impatient reader is invited to 
peek ahead to Assumption \ref{ass:ergodicity}).
\item Working with multidimensional stock-prices processes is of
  fundamental importance for the understanding of the full scope of
  our results.  In order to simplify the presentation, we introduce
  several notational shortcuts for ordinary and stochastic integrals
  of vector- or matrix-valued processes; for an integrable
  $\R^m$-valued process $\bsy{\rho}(t)=\vnot{\rho_i(t)}{n}$, and a
  sufficiently regular $\R^m$-valued process
  $\bpi(t)=\vnotr{\pi_j(t)}{m}$ we write
\begin{equation}
    \nonumber 
    \begin{split}
      \int_0^t \bsy{\rho}(u)\, du\triangleq \sum_{i=1}^n \int_0^t
      \rho_i(u)\, dt,\quad 
      \int_0^t \bpi(t)\, d\bW(t)\triangleq \sum_{j=1}^m \int_0^t
      \pi_j(t)\, dW_j(t). 
    \end{split}
\end{equation}
\end{enumerate}
\end{remark}

\subsection{Trading strategies and wealth}
Actions of an investor in the market are modelled by the proportions of
current wealth her/she invests in various assets. Specifically, we
have the following formal definition.
\begin{definition}
\label{def:portfolio-proportions}
An
$\R^n$-valued stochastic process
$\cprfi{\bze(t)}$ $=\cprfi{\vnot{\zeta_i(t)}{n}}$ is called an {\em
  admissible portfolio-proportion  process} if 
it is progressively measurable and
satisfies 
\begin{equation}%
\label{equ:regularity-zeta}
    \begin{split}
      \int_0^t \abs{\bzt(u)(\balpha(u)-r\bone)}\, du+\int_0^t
      \norm{\bzt(t) \bsi(u)}^2\, du<\infty, \text{ a.s., for all }
      t\in [0,\infty),
    \end{split}
\end{equation} 
where, as usual, $\bone=(1,\dots,1)^T$ is an $n$-dimensional column
vector all of whose coordinates are equal to $1$, and
$\norm{\bsy{x}}=\sqrt{\sum_{j=1}^m x_j^2}$ is the standard Euclidean
norm of a vector $\bsy{x}=\vnotr{x_j}{m}\in\R^m$.
\end{definition}
Given a portfolio-proportion process $\bze(t)$, we interpret its $n$
coordinates as the proportions of the current wealth $\Xze(t)$
invested in each of $n$ stocks. In order to remain self-financing, the
left-over wealth $\Xze(t)(1-\sum_{i=1}^n \zeta_i(t))$ is assumed to be
invested in the riskless bond $S_0(t)$.  Of course, if this quantity
is negative, we are effectively borrowing at the rate $r>0$. We stress
that no short-selling restrictions are imposed, meaning that the
proportions $\zeta_i(t)$ are allowed to be negative.  Therefore, the
equation governing the evolution of the total wealth $\cprfi{\Xze(t)}$
of the investor using the portfolio-proportion process
$\cprfi{\bze(t)}$ is given by
\begin{equation}%
\label{equ:wealth-one}
    \begin{split}
d\Xze(t)&= \Xze(t) \Big( 
\bzt(t) \balpha(t)\, dt+\bzt(t)\bsi(t)\, d\bW(t) \Big)
+\Big(1-\bzt(t)\bone\Big) \Xze(t)r \, dt\\
&= \Xze(t) \Big( (r+\bzt(t) \bmu(t))\, dt+\bzt(t)\bsi(t)\,
d\bW(t)
\Big),
    \end{split}
\end{equation}
where $\cprfi{\bmu(t)}=\cprfi{\vnot{\mu_i(t)}{n}}$, with
$\mu_i(t)=\alpha_i(t)-r$ for $i=\ft{1}{n}$, is the vector of {\em
  excess rates of return}.  Under
regularity conditions \eqref{equ:regularity-zeta} imposed on $\bze(t)$
above, \eqref{equ:wealth-one} admits a unique strong solution given by
the explicit expression
\begin{equation}%
\label{equ:wealth-two}
    \begin{split}
      \Xze(t)=X(0)\exp\left( \int_0^t
        \Big(r+\bzt(u)\bmu(u)-\tot\norm{\bzt(u)\bsi(u)}^2\Big)\, du+
        \int_0^t \bzt(u) \bsi(u)\, d\bW(u) \right),
    \end{split}
\end{equation}
The initial wealth $\Xze(0)=X(0)\in (0,\infty)$, is
considered a primitive of the model, and will thus be considered
 arbitrary but fixed throughout the paper. In particular, it will not
 vary with the choice of the investment strategy $\bze$. 
\subsection{Some useful notation}
\subsubsection{Functions $\tilde{Q}$ and $Q$} 
The expression appearing inside the first integral in
\eqref{equ:wealth-two} above will be important enough in the sequel to
warrant its own notation; the affine-quadratic function
$\tilde{Q}:\R^2\to\R$ is defined as
\begin{equation}%
\label{equ:definition-of-Q}
    \begin{split}
      \tilde{Q}(\zmu,\zsigma)= r+\zmu-\tot \zsigma^2,
    \end{split}
\end{equation}
so that the aforementioned expression becomes
$\tilde{Q}(\bzt(t)\bmu(t), \norm{\bzt(t)\bsi(t)})$. Oftentimes, the
dependence of the drift of the process $\log(\Xze(t))$ on the choice
of the instantaneous portfolio-proportion $\bze$ will be important. It
proves useful to define the random field $Q:\Omega\times
[0,\infty)\times \R^n\to\R$ by
\begin{equation}%
\label{equ:random-field-Q}
    \begin{split}
      Q(t,\bze)=\tilde{Q}(\bzt \bmu(t), \norm{\bzt\bsi(t)}).
    \end{split}
\end{equation}
In the new notation, the process $\Yze(t)=\log(\Xze(t))$ evolves
according to the following simple dynamics
\begin{equation}%
\label{equ:dynamics-of-Y}
    \begin{split}
d\Yze(t)= Q(t,\bze(t))\, dt+ d\Mze(t),\ t\in [0,\infty),
    \end{split}
\end{equation}
where the local martingale $\cprfi{\Mze(t)}$ is defined by
\begin{equation}%
\label{equ:defined-Mze}
    \begin{split}
\Mze(t)=\int_0^t \bzt(u)\bsi(u)\, d\bW(u),\ t\in [0,\infty).
    \end{split}
\end{equation}
It is clear from the expression \eqref{equ:wealth-two} above that
the drift of $\Xze(t)$ depends on the $\R^n$-dimensional process
$\bze(t)$ only through two ``sufficient statistics'' 
\begin{equation}%
\label{equ:portfolio-quantities}
    \begin{split}
 \zmu(t)\triangleq \bzt(t)\bmu(t),\text{ and }\zsigma(t)\triangleq
\norm{\bzt(t)\bsi(t)}.
    \end{split}
\end{equation}
They will be referred to in the sequel as 
{\bf portfolio rate of return} and 
{\bf portfolio volatility}, respectively. 

\subsubsection{The Merton-proportion process}
In order for the definition of the Merton-proportion process to make sense,
we impose the following mild condition on the variance-covariance
process $\bsi(t)$.
\begin{assumption}
\label{ass:independent-rows}
The matrix $\bsi(t)$ has independent rows for all $t\in [0,\infty)$, a.s.
\end{assumption}
The financial meaning of Assumption \ref{ass:independent-rows} is
quite simple - it precludes different stocks from having the same
diffusion structure. Otherwise, the market would either allow for 
 arbitrage opportunities or  redundant assets would exist. 
The first consequence of this assumption is that $n\leq m$ - the
number of risky assets does not exceed the number of ``sources of
uncertainty''. Also, the inverse $(\bsi(t)\bsi^T(t))^{-1}$ is easily
seen to exist and so the equation
\begin{equation}%
\label{equ:def-Merton}
    \begin{split}
\bsi(t) \bsi^T(t) \bze_M(t)= \bmu(t),
    \end{split}
\end{equation}
uniquely defines a c\' agl\' ad
stochastic process $\cprfi{\bze_M(t)}$, termed the 
{\em Merton-proportion process}. It has the pleasant
property that (in the absence of portfolio constraints), the
growth-rate- or $\log$-optimizing investor would invest in the
market exactly using the components of $\bze_M(t)$ as portfolio
proportions (see \cite{KarShr98}). 

\subsubsection{A metric-valued process} Finally, the fact that the rows of $\bsi(t)$ are independent easily leads to
the fact that the random field  $\cprfi{d_{\bsi(t)}(\cdot,\cdot)}$
given by
\begin{equation}%
\label{equ:dsig-defined}
    \begin{split}
d_{\bsi(t)}(\bze_1,\bze_2)=
\norm{\bsi(t)^T(\bze_1-\bze_2)},\text{  for }\bze_1, \bze_2\in\R^n.
    \end{split}
\end{equation}
is a metric on $\R^n$ for all $t\in [0,\infty)$, a.s.

\subsection{The Ergodic Assumption}
Since we are dealing with a stochastic control problem of an
ergodic type, we impose an ergodicity requirement on
the coefficients $\bmu(t)$ and $\bsi(t)$ driving the financial
market. It turns out, somewhat surprisingly, that we only need to
deal with a  combination 
of two - a real valued process related to the Merton-proportion process
defined above in \eqref{equ:def-Merton}. 
Specifically, we impose  the following assumption.
\begin{assumption}[Ergodicity of the Merton-proportion process]\ 
\label{ass:ergodicity}
The process $\cprfi{\norm{\bzt_M(t)\bsi(t)}}$
is ergodic in the sense that for each non-negative
continuous function $\vp:[0,\infty)\to\R$ satisfying $\sup_{x\in\R}
\frac{\vp(x)}{1+x^2}<\infty$ there exists an
$\FF_{\infty}$-measurable, finite random variable $Z(\vp)$ such that
\begin{equation}%
\label{equ:ergodic-limit}
    \begin{split}
      \lim_{t\to\infty} \frac{1}{t}\int_0^t
      \vp(\norm{\bzt_M(u)\bsi(u)})\, du=Z(\vp),\ \text{a.s.}
    \end{split}
\end{equation}
\end{assumption}
\begin{remark}
Multiplying both sides of \eqref{equ:def-Merton} by $\bze_M^T(t)$
from the left shows that $\norm{\bze_M^T(t)\bsi(t)}^2=\bze_M^T (t)
\bmu(t)$, so the Assumption \ref{ass:ergodicity} can be formulated
equivalently in terms of the process $\bze_M^T(t)\bmu(t)$. This fact
will be useful in some proofs in the sequel.
\end{remark}

\begin{example}
\label{exa:ergodicity}
Assumption \ref{ass:ergodicity} is the only non-trivial
condition imposed on the market coefficients $\bmu(t)$ and $\bsi(t)$,
and, therefore, examples to illustrate its restrictiveness are needed.
Two important classes of financial models satisfying it are presented
below:
\begin{enumerate}
\item When $\bmu(t)$ and $\bsi(t)$ are deterministic constants one
  can easily see that Assumption \ref{ass:ergodicity} is trivially
  satisfied. More generally, our framework incorporated deterministic
  processes $\bmu(t)$ and $\bsi(t)$ which exhibit enough periodic
  behavior in order for the averages introduced in
  \eqref{equ:ergodic-limit} to be convergent. Deterministic
  coefficients of this type are used in models of seasonally-sensitive
  assets.
\item A class of stochastic-volatility models also complies with
  Assumption \ref{ass:ergodicity}.  Indeed, following \cite{FleHer03}
  and \cite{FouTul02}, let us consider the special case of our model
  in which $n=1$, $m=2$ and
\begin{equation}%
\label{equ:stochastic-volatility-coefficients}
    \nonumber 
    \begin{split}
\mu_1(t)=\mu\in\R,\ \sigma_{11}(t)=\rho\, \Sigma(V(t)),
\text{ and }\sigma_{12}(t)=\sqrt{1-\rho^2}\, \Sigma(V(t)), 
    \end{split}
\end{equation}
where, the state process $\cprfi{V(t)}$ is given by
\begin{equation}
\label{equ:vol-of-vol}
    \begin{split}
dV(t)=\nu (\overline{V}-V(t))\,dt+\,d{W_2}(t).
    \end{split}
\end{equation}
Here $\rho\in [-1,1]$ is the correlation coefficient, $\nu>0$ is
the rate of mean reversion, $\overline{V}>0$ is the mean-reversion
level, $\mu\in\R$ is the mean rate of return of the risky asset, and
the function $\Sigma:[0,\infty)\to [0,\infty)$ transforms the state
process $v(t)$ into asset volatility $\Sigma(V_t)$.  We assume that
$\Sigma$ is a continuous function bounded both from above and away
from zero. Under these conditions, the volatility $\Sigma(V(t))$
inherits the mean-reversion property of $V(t)$, reverting to level
$\Sigma(\overline{V})$.  
In this model $\norm{\bze^T_M(t)\bsi(t)}^{2}=\mu^2 \Sigma(V(t))^{-2}=g(V_t)$
for some bounded continuous function $g:[0,\infty)\to\R$.
The processes $V_t$, being a one-dimensional Ornstein-Uhlenbeck process,
has a finite invariant measure $\gamma(dx)$ which is Gaussian. In
fact,
$\gamma(dx)=\sqrt{\frac{\nu}{\pi}}e^{-\nu(x-\overline{V})^2}dx.$ By
$Theorem~3.1$ in \cite{Kha60}, for any measurable function $\vp$
integrable with respect to $\gamma(dx)$, we have
\begin{equation}
    \nonumber 
    \begin{split}
\lim_{t\to\infty}\frac{1}{t}\int_{0}^{t}\vp(V(u))\,du=
\int_{-\infty}^{\infty}\vp(x)\gamma(dx)<\infty. 
    \end{split}
\end{equation}
Therefore,  Assumption \ref{ass:ergodicity} holds, with
$Z(\vp)=\int_{-\infty}^{\infty} \vp(x)\gamma(x)\, dx$.
\end{enumerate}  
\end{example}
\revision{
\begin{remark}\ 
\begin{enumerate}
\item\revision{
The ergodicity assumption stated above is standard in control problems
with ergodic objective criterion. A typical ergodic process used in
practice is very much like the one appearing in the
stochastic-volatility model in the Example \ref{exa:ergodicity}, (2)
above --- a deterministic function of a diffusion process with a
stationary distribution. While the financial models of the 
asset prices 
will
not have the ergodic property, it is hard to immagine realistic
models of appreciation rates or the volatility matrix which are {\em
  not} ergodic. Indeed, 
from the economic point of view, the lack of the ergodic structure in those
processes would imply strong confidence of the modeller 
in the lack of any kind of equilibrium in the very-long term behaviour
of the financial system under consideration. }
\item \revision{The random variable $Z(\vp)$ is, in fact, a deterministic constant
 throughout Example \ref{exa:ergodicity} above. There is, however, a class of
 realistic situations in which it will be a true random
 variable. Imagine a situation in which, from the start, the market
 volatility is known to be as in part (2) of Example
 \ref{exa:ergodicity}, but the mean-reversion level $\bar{V}$ is
 unknown, with the a-priory distribution $\nu(dx)$, which is
 independent of all the other sources of uncertainty in the model. In
 that case, the random variable $Z(\vp)$ will be truly random (except
 for special choice of the function $f$) and given by 
\[ Z(\vp)(\omega)=\int_{-\infty}^{\infty} \vp(x) 
\sqrt{\frac{\delta}{\pi}}e^{-\delta(x-\overline{V}(\omega))^2} dx.\]
More complicated cases where $\bar{V}$ is not independent of the
other sources of uncertainty can be envisioned. Those situations will
lead to random limits $Z(\vp)$, but will not correspond to the simple mixtures
of different stochastic volatility models any more.}
\end{enumerate}
\end{remark}
}

\revision{
\subsection{Portfolio constraints}
\label{sub:portfolio-constraints}
Having introduced the financial market, we turn to the
specification of the portfolio constraints which limit the investor's
behavior in each instant. We start with an abstract description  of the
form of these constraints and continue to present several special
cases dealing with realistic risk-limits.  
One of the features in which our framework differs from the majority
of existing work is that the set of
allowable portfolio proportions  depends not only  on the current market
conditions ($\bmu(t)$ and $\bsi(t))$ but also on the current level of
the investor's wealth $\Xze(t)$.
\begin{definition}
\label{def:port-corr}
\revision{
A {\bf portfolio-constraint correspondence} is a family of
$(x,\bmu,\bsi)\mapsto F(x,\bmu,\bsi)\subseteq \R^m$ of subsets of
$\R^n$ with the property that there exist two functions 
$f:\R\times [0,\infty)\to\R\cup\set{\infty}$  and
$h:(0,\infty)\to\R$ such that
\[ 
F(x,\bmu,\bsi)=F_{(f,h)}(x,\bmu,\bsi)=\sets{\bze\in\R^m}{f(\bzt\bmu, 
\norm{\bzt\bsi})\leq h(x)}.
\]
The function $f$ is assumed to satisfy the following conditions:
\begin{enumerate}
\item \revision{ $f\in C^1(\R\times [0,\infty))$ is jointly convex.}
\item \revision{
For each $(\zmu,\zsigma)\in \R\times [0,\infty)$, the sections
  $f(\zmu,\cdot)$ and $f(\cdot,\zsigma)$ are (respectively)
  strictly increasing and decreasing.}
\item \label{ite:kappas} \revision{
$f(0,0)<0$ and there exist constants
  $\kappa_i>0$, $i\in\set{1,2,3}$ such that for all
  $(\zmu,\zsigma)\in\R\times [0,\infty)$
\begin{equation}%
\label{equ:lower-bound-on-f}
    \begin{split}
 f(\zmu,\zsigma)\geq \kappa_1 \zsigma^2-\kappa_2 \zmu-\kappa_3
    \end{split}
\end{equation}}
\end{enumerate}
\revision{For the function $h$, we require either of the following two sets of
assumptions:}
\begin{itemize}
\item[(A)] \revision{$h(x)=c$ for some $c\in(0,\infty)$ and  all $x>0$, or}
\item[(B)] \ 
\begin{itemize}
  \item[(B.1)] \revision{ $h(\exp(\cdot))$ is convex,}
  \item[(B.2)] \label{ite:abs-x0} 
\revision{ there exists $x_0>0$ such that
    $h(x)=+\infty$ for $x\leq x_0$,}
  \item[(B.3)] \revision{ $h(\cdot)$ is finite, strictly decreasing and
    continuously differentiable on $(x_0,\infty)$, and}
  \item[(B.4)] \revision{ $\lim_{x\searrow x_0} h(x)=+\infty$,
    $\lim_{x\to\infty} h(x)=0$.}
  \end{itemize}
\end{itemize}
\revision{
The constraints are said to be {\bf relative} in the case (A), and
{\bf absolute} if (B) holds. } }
\end{definition}
}

In notational analogy with the function $\tilde{Q}$ and the random
field $Q$ (defined in \eqref{equ:definition-of-Q} and
\eqref{equ:random-field-Q}), for each portfolio-constraint
correspondence $F=F_{(f,h)}$ as in Definition \ref{def:port-corr}, 
 we define a random set-valued
field (random correspondence field) $\tilde{F}:\Omega\times [0,\infty)
\times\R\to 2^{\R^n}$ by
\begin{equation}%
\label{equ:definition-correspondence-F}
    \begin{split}
\tilde{F}(t,x)=\tilde{F}_{(f,h)}(t,x)=F(x,\bmu(t),\bsi(t)).
    \end{split}
\end{equation}
This parallel notation will be very useful in the later sections of
the manuscript. 

Imposing a portfolio constraint dynamically leads to the following definition
of the set of admissible portfolio processes.
\begin{definition}
\label{def:var-admissible}
An $\R^n$-valued process $\cprfi{\bze(t)}$
is said to be {\bf $(f,h)$-admissible} if it is an admissible
 portfolio-proportion
process (in the sense of Definition \ref{def:portfolio-proportions}) and
\begin{equation}
    \nonumber 
    \begin{split}
\bze(t)\in F_{(f,h)}(t,\Xze(t))=\tilde{F}_{(f,h)}
(\Xze(t), \bmu(t),\bsi(t)),\text{
  for all $t\in [0,\infty)$, a.s.},
    \end{split}
\end{equation}
where the dynamics of the process $\cprfi{\Xze(t)}$ is given in
\eqref{equ:wealth-one} and \eqref{equ:wealth-two}.
The set of all admissible portfolio-proportion processes
$\cprfi{\bze(t)}$ will be denoted by $\AA_{(f,h)}$, or simply $\AA$,
when no confusion can arise. 
\end{definition}

\revision{
\subsection{Examples of portfolio constraints} 
The discussion that follows aims to show that a number of risk-based
portfolio constraints used in the literature allows a formulation from
Definition \ref{def:port-corr}. 
}
\subsubsection{Projected  distribution of wealth}
\label{sss:projected-wealth}
  For the purposes of
risk measurement, it is a common practice to use an approximation
of the distribution of the investor's wealth at a future date. 
Given a  fixed time-instance $t_0\geq 0$, and a length $\tau>0$ 
of the measurement horizon $[t_0,t_0+\tau]$, 
the {\em projected distribution} of the wealth from trading is usually
calculated under the simplifying assumptions that
\begin{enumerate}
\item the proportions of the wealth $\{\bze(t)\}_{t\in
    [t_0,t_0+\tau]}$  invested in various
  securities, as well as
\item the market coefficients $\{\balpha(t)\}_{t\in [t_0,t_0+\tau]}$
  and $\{\bsi(t)\}_{t\in [t_0,t_0+\tau]}$
\end{enumerate}
will stay constant and equal
to their present values throughout the time interval $[t_0,t_0+\tau]$.
The wealth equations \eqref{equ:wealth-one} and
\eqref{equ:wealth-two} yield that the {\bf projected wealth loss} is
- conditionally on $\FF_{t_0}$ - 
distributed as $L=L(X(t_0),
\zmu(t_0),\zsigma(t_0))$, where the law of $L(x,\zmu,\zsigma)$
is the one of 
\begin{equation}%
\label{equ:law-of-L}
    \begin{split}
x \Big(1-\exp (Y(\zmu,\zsigma))\Big),
    \end{split}
\end{equation}
in which $Y(\zmu,\zsigma)$ is a normal random variable 
with mean $\tilde{Q}(\zmu,\zsigma)\tau$ and the 
standard deviation $\sqrt{\tau}\zsigma$. The quantities $\zmu(t_0)$
and $\zsigma(t_0)$ are the portfolio rate of return and volatility, 
defined in \eqref{equ:portfolio-quantities}.
\revision{
\begin{remark}
  The notion of the projected distribution of wealth as defined above
  has first appeared in the financial literature in \cite{CuoHeIss07}
  in the context of constant coefficients. As one of the referees
  points out, while it is reasonable to keep the portfolio proportions
  constant throughout the measurement horizon $[t_0,t_0+\tau]$, the
  same cannot be said about the constancy of the market-coefficients
  $\balpha(\cdot)$ and $\bsi(\cdot)$. Indeed, under the conditions
  encountered in financial practice, the random evolution of $\balpha$
  and $\bsi$ will typically lead to a more dispersed wealth
  distribution, and, consequently, to an under-estimate of the
  riskiness of the current position. The reason for such an assumption
  is the fact that it leads to a log-normally distributed wealth,
  which, in turn, greatly simplifies the analysis and leads to
  explicit form of the optimal policies. A simple and practically
  implementable way out of this predicament is to retain the
  assumption of the constancy of the market coefficients throughout
  the measurement horizon, but to use a ``corrected'' versions
  $\check{\bsi}$ and $\check{\balpha}$ of the current values
  $\bsi(t_0)$ and $\balpha(t_0)$ of  the processes $\bsi$ and
  $\balpha$. These, corrected, versions should correspond to a normal
  approximation of the true distribution of the investor's wealth at
  time $t_0+\tau$, and can be obtained in closed from in many of the
  models used in practice (see \cite{SteSte91} for the case of
  stochastic volatility from Example \ref{exa:ergodicity} (2)).  As
  the reader can easily check, such a corrected distribution will
  still lead to a portfolio-constraint compliant with Definition
  \ref{def:port-corr}. In the case when processes driving the market
  coefficients are Markovian, a closed-form expression for the
  distribution of the wealth at time $t_0+\tau$ is available, and the
  conditions in the Definition \ref{def:port-corr} can be checked, one
  can use this exact distribution instead of the projected one. The
  authors have been unable, however, to identify any interesting cases
  where such a procedure is possible. Moreover, we feel that the
  approximation approach described above is more feasible for the
  practical application for yet another reason: even if one is able to
  identify the exact distribution of the wealth at time $t_0+\tau$,
  one still faces the much more difficult problem of estimation of the
  coefficients $\balpha(t_0)$ and $\bsi(t_0)$. We leave the implementation
  of a practical solution to this serious predicament for future
  research.
\end{remark}
}

\subsubsection{Risk limits}
\label{sub:risk-limits}
 The purpose of this subsection is to
define and expose certain properties of the risk measures ($\var$, $\tvar$  and
$\lel$) discussed in the Introduction. Each one of these will be
 introduced
through a family of random sets depending on the present values of the
market coefficients, just like the ones in Definition
\ref{def:port-corr}.
 Put differently, our three risk measures will
give rise to a random, wealth-dependent portfolio constraints. 
Strictly speaking, $\var$, $\tvar$  and $\lel$ define {\em families} of risk
measures, parameterized by exogenously chosen percentile parameter
$\alpha$,
as well as the risk constraint parameters $\aa_V, \aa_T,\aa_L >0$ and
$\ar_V, \ar_T, \ar_L\in (0,1)$. We
will assume that $\alpha$ is fixed and constant and that it satisfies 
$\alpha \in (0,1/2)$. This technical 
assumption relates well to the practice
where the typical values of $\alpha=0.05$, or $\alpha=0.1$ are used.
It will be assumed through the rest of the paper that these parameters
are arbitrarily chosen and fixed. Together with the market
coefficients and the measurement horizon $\tau$, they will play the
role of ``global variables''. 
\begin{definition}
\label{def:var}
The {\bf value-at-risk} $\var=\var(x,\zmu,\zsigma)$ -
  corresponding to the current wealth $x$, the 
 portfolio rate of return $\zmu$
  and volatility $\zsigma$ - is the positive part of the upper 
$\alpha$-percentile of the
  projected loss distribution $L=L(x,\zmu,\zsigma)$,
  i.e.,
\begin{equation}
    \nonumber 
    \begin{split}
\var=\gamma_{\alpha}^+=\max(0,\gamma_{\alpha}),\text{ where
  $\gamma_{\alpha}$ uniquely satisfies } 
\PP[ L \geq \gamma_{\alpha} ]=\alpha.
    \end{split}
\end{equation}
\end{definition}
\begin{definition}
\label{def:tvar}
The {\bf tail value-at-risk} $\tvar=\tvar(x,\zmu,\zsigma)$  is
the positive part of the mean
  of the distribution of the projected loss distribution,
 conditioned to take a value above its upper $\alpha$-percentile, i.e.,
\begin{equation}%
\label{equ:tvar-def}
    \nonumber 
    \begin{split}
      \tvar=w_{\alpha}^+,\text{ where $\gamma_{\alpha}$ satisfies }
      \PP[ L \geq \gamma_{\alpha} ]=\alpha,\text{ and }
      w_{\alpha}=\EE[ L | L\geq \gamma_{\alpha}].
    \end{split}
\end{equation}
\end{definition}
Our third measure of risk - $\lel$ - is similar to $\tvar$, with one
significant difference: it does not take the market rate-of-return in
consideration. More precisely, we have the following definition
\begin{definition}
\label{def:lel}
 The {\bf limited expected loss} $\lel=\lel(x,\zsigma)$ is the 
tail value-of-risk corresponding to the 
  loss distribution  $L=L(x,0,\zsigma)$ in
  which the portfolio rate of return is set to $0$.
\end{definition}
\begin{remark}\ 
\begin{enumerate}
\item
In the common case when the financial market admits an equivalent 
 martingale measure $\QQ$, $\lel$ can be interpreted as the $\tvar$
 calculated under $\QQ$. The reader will easily convince him- or
 herself that, within our modelling 
framework at least, $\lel$
 will not depend on the choice of $\QQ$, should there exist more than
 one.
\item \revision{Definitions \ref{def:var} and \ref{def:tvar} differ
    slightly from the definitions of the Value at Risk and Tail Value
    at Risk given in \cite{FolSch04}: positive parts (not present in 
\cite{FolSch04}) are introduced in order to penalize only
losses. Otherwise, it could happen that the induced constraints would,  
effectively, require the investor to make a certain, positive, return.  } 
\end{enumerate} 
\end{remark}
\subsubsection{Relative versions of risk measures}
All three $\var$, $\tvar$ and $\lel$ measure the risk of a large
loss in {absolute terms}. If we define the {\em relative projected wealth
loss} as the distribution of the  positive quantity
$\frac{\Xze(t_0)-\Xze(t_0+\tau)}{\Xze(t_0)}$ (under the simplifying
assumptions 1.~and 2.~from paragraph \ref{sss:projected-wealth}
above), 
definitions of the
analogous relative quantities $\var_r$, $\tvar_r$ and $\lel_r$ can
readily be given. In fact, due to the multiplicative structure of the
wealth equations \eqref{equ:wealth-one} and \eqref{equ:wealth-two}, we
have the following expressions
\begin{equation}%
\label{equ:relative-vars}
    \begin{split}
 \var_r(\zmu,\zsigma)=\frac{\var(x,\zmu,\zsigma)}{x}, &\qquad 
 \tvar_r(\zmu,\zsigma)=\frac{\tvar(x,\zmu,\zsigma)}{x},\text{
   and } \\
 \lel_r(\zmu,\zsigma)&=\frac{\lel(x,\zmu,\zsigma)}{x}.
    \end{split}
\end{equation}
As we would expect, the relative risk limits $\var_r$, $\tvar_r$ and
$\lel_r$ no longer depend on the current level of wealth $x$. 
\subsubsection{Some explicit expressions}
Thanks to the fact that the distribution appearing in
\eqref{equ:law-of-L} is normal, explicit formulae can be given for the
values of all three risk measures appearing above.
\begin{proposition}
\label{pro:formulas-for-risk-measures}
For $\zmu\in\R$ and $\zsigma>0$, we have 
\begin{align}
\label{equ:var}      
\var(x,\zmu,\zsigma)&= x \left[
      1-\exp\Big(\tilde{Q}(\zmu,\zsigma)\tau+N^{-1}(\alpha)\zsigma\sqrt{\tau}\Big)
      \right]^+\\
\label{equ:tvar}
      \tvar(x,\zmu,\zsigma)&= x \left[
      1-\tfrac{1}{\alpha}e^{\tau(r+\zmu)}
      N(N^{-1}(\alpha)-\zsigma\sqrt{\tau})
      \right]^+,\text{ and }\\
\label{equ:lel}
      \lel(x,\zsigma)&= x \left[
      1-\tfrac{1}{\alpha}e^{r\tau}
      N\Big( 
      N^{-1}(\alpha)-\zsigma\sqrt{\tau}\Big)
      \right]^+,
\end{align}
where $N:\R\to (0,1)$ is the cumulative distribution function of a
standard normal random variable.
\end{proposition}

\begin{proof}
See Appendix \ref{sec:technical-results}.
\end{proof}
\subsubsection{Constraints corresponding to risk measures} For
 a constant $x>0$, a vector $\bmu\in\R^n$
and a matrix $\bsi\in\R^{n\times m}$, define
\begin{equation}%
\label{equ:var-constraint-set}
    \begin{split}
\tfa_V(x,\bmu,\bsi)=
\sets{\bze\in\R^n}{\var(x,\bzt\bmu,
  \norm{\bzt\bsi})\leq \aa_V},
    \end{split}
\end{equation}
where $\aa_V>0$ is an exogenously defined constant.
In words, $\tfa_V$ is the set of all portfolio proportion vectors
$\bze\in\R^n$ such that the loss incurred by keeping a fixed portfolio
proportions $\bze$ in a market with constant rate-of-return $\bmu$ and
$\bsi$, over a time horizon $[t,t+\tau]$ and with current wealth $x$
results in no violation of the $\var$ risk limit. \revision{Thanks to
Proposition \ref{pro:formulas-for-risk-measures}, one can check that
the correspondence $\tfa_V$ from above is, in fact, a special case of
a correspondence $F_{(f,h)}$ from Definition \ref{def:port-corr} (see
Appendix \ref{sec:technical-results}.)}

\label{sss:other-risk-limits}
Using the other 5 risk measures ($\tvar, \lel, \var_r, \tvar_r$ and
$\lel_r$) constraint sets $\tfa_T$, $\tfa_L$, $\tfr_V$, $\tfr_T$ and
$\tfr_L$, as well as their random-correspondence versions $\fa_T$,
$\fa_L$, $\fr_V$, $\fr_T$ and $\fr_L$ (constructed as in 
\eqref{equ:definition-correspondence-F}) can be defined as in
\eqref{equ:var-constraint-set}.   One should note that for the
relative versions  the dependence on the
current wealth level $\Xze(t)$ is lost, and the portfolio constraint
set will depend on the values of the market coefficients
only. One can easily check that for these, relative, versions,
functions $f$ and $h$ can be chosen so that the
  condition (B) in Definition \ref{def:port-corr} is satisfied. For
  the absolute versions, on the other hand,  the condition (A) will be met.

\subsection{The optimization problem}
We finish the section by the formulation of our central problem. 
Given a choice of the constraint $\AA=\AA_{(f,h)}$ as in Definition
\ref{def:port-corr},  we are
searching for a portfolio-proportion process $\optbz(t)\in\AA$
such that, for all $\bze(t)\in\AA$,  
\begin{equation}
    \nonumber 
    \begin{split}
\liminf_{t\to\infty} \frac{\log(\Xoze(t))}{t}\geq 
\liminf_{t\to\infty} \frac{\log(\Xze(t))}{t},\text{ a.s.}
    \end{split}
\end{equation}  
\begin{remark}
  While it is intimately related to the problem of maximizing
  logarithmic utility $\EE[\log(\Xze(T))]$, the ergodic problem we
  address here differs considerably from it - mainly in its dependence
  on the ergodicity of the market coefficients. \revision{
The limiting nature of
  the objective criterion corresponds to a long-term average of the
  underlying controlled stochastic processes, while the classical
  logarithmic utility can only be applied to a finite (and fixed) time
  horizons. Moreover, the dependence
  of the constraint set on the current wealth (in the absolute case)
  rules out the na\" ive myopic approach characteristic of the
  behavior of a logarithmic investors on finite
  horizons.} The relative case is much simpler and we can, in fact,
  treat both logarithmic and ergodic growth problems on the same
  footing (see the second part of the Main Theorem \ref{thm:main}).
\end{remark}

\section{Main results} Our main result - Theorem \ref{thm:main} -
summarizes the central findings of the manuscript. Its proof is the
content of Section \ref{sec:analysis}, below. 
\begin{theorem}
\label{thm:main}
Let the financial market $\cprfi{S_0(t),S_1(t),\dots, S_n(t)}$ be
defined as in \eqref{equ:SDSs-for-stocks}, with the coefficients
$r>0$, $\cprfi{\balpha(t)}$ and $\cprfi{\bsi(t)}$ satisfying
Assumptions \ref{ass:basic-regularity}, \ref{ass:independent-rows} and
\ref{ass:ergodicity}.  \revision{Furthermore, let the functions 
$f$ and $h$, as well as the corresponding admissible class
$\AA=\AA_{(f,h)}$ be as in Definition \ref{def:port-corr}.}
Then the following statements hold.
\begin{enumerate}
\item {\bf Absolute constraints}\\\noindent Suppose that the function
  $h$ satisfies the assumption set (A) from Definition
  \ref{def:port-corr}.
 Let
  $\cprfi{\bze_M(t)}$ be the {\em Merton-proportion} process defined
  in \eqref{equ:def-Merton}.  There exists a stochastic processes
  $\cprfi{\bs(t)}$ and $\cprfi{\bi(t)}$ taking values in $(0,1]$ such
  that the vector-valued processes $\cprfi{\bzs(t)}$ and
  $\cprfi{\bzi(t)}$, defined by
\begin{equation}%
\label{equ:define-bzs-and-bzi}
    \begin{split}
 \bzs(t)=\bs(t) \bze_M(t),\ \bzi(t)=\bi(t) \bze_M(t), 
    \end{split}
\end{equation}
have the following properties
\begin{enumerate}
\item both $\cprfi{\bzs(t)}$ and $\cprfi{\bzi(t)}$ are c\' agl\' ad and
define strictly  positive wealth processes
\begin{equation}
    \nonumber 
    \begin{split}
\qquad\Xzs(t)&= X(0)\exp\Big( \int_0^t Q(t,\bzs(t))\,
      dt+ \int_0^t (\bzs(t))^T \bsi(t)\, d\bW(t) \Big), \\
\qquad\Xzi(t)&= X(0)\exp\Big( \int_0^t Q(t,\bzi(t))\,
      dt+ \int_0^t (\bzi(t))^T \bsi(t)\, d\bW(t) \Big), \\
    \end{split}
\end{equation}
\item $\bzi(t)$ is the unique projection of $\bze_M(t)$ onto the
   limiting constraint set 
\begin{equation}
    \nonumber 
    \begin{split}
F(t,\infty)=F_{(f,h)}(t,\infty)=\cap_{x>0} F(t,x),
\end{split}
\end{equation}
under the metric $d_{\bsi(t)}$ on $\R^n$ defined by \eqref{equ:dsig-defined}. 
\item $\bzs(t)$ is the unique $d_{\bsi(t)}$-projection 
of $\bze_M(t)$ onto the  
constraint set $F(t,\Xzs(t))$.
\item $\cprfi{\bzs(t)}$ and $\cprfi{\bzi(t)}$ are $(f,h)$-admissible
  and 
\begin{equation}
    \nonumber 
    \begin{split}
\lim_{t \to \infty} \frac{\log(\Xzs(t))}{t}=
\lim_{t \to \infty} \frac{\log(\Xzi(t))}{t}=r+Z(x^2 \delta(x)),
    \end{split}
\end{equation}
where $Z(\cdot)$ is the random variable introduced in Assumption
\ref{ass:ergodicity}, and $\delta:[0,\infty)\to (0,1]$ 
is a non-negative continuous function depending only on the constraint
type, but independent of the market coefficients. 
\item \revision{More precisely, for
$\ld\geq 0$, $\delta(\ld)=\min(g(\ld),1)$, where $g(\ld)$ 
is the unique positive solution of the
equation
\begin{equation}%
\label{equ:equat-for-delta}
    \begin{split}
f( g(\ld) \ld^2, g(\ld) \ld)=0.
    \end{split}
\end{equation} 
Additionally, with
 $\delta^*(\ld,x)=\min(g^*(\ld,x),1)$, where $g^*(\ld,x)$ is 
is the unique positive solution of the
equation
\begin{equation}%
\label{equ:equat-for-delta-star}
    \begin{split}
f( g^*(\ld,x) \ld^2, g^*(\ld,x) \ld)=h(x),\ x,\ld>0, 
    \end{split}
\end{equation}  
we have
\begin{equation}%
\label{equ:explicit}
    \begin{split}
\bzs(t)=\delta^*(\norm{\bze_M(t) \bsi(t)},\Xzs(t)) \text{ and }
\bzi(t)=\delta(\norm{\bze_M(t) \bsi(t)}).
    \end{split}
\end{equation}
}

\item Both $\bzs(t)$ and $\bzi(t)$ are growth optimal in the sense that
\begin{equation}
    \nonumber 
    \begin{split}
\liminf_{t \to \infty} \frac{\log(\Xze(t))}{t}\leq 
\lim_{t \to \infty} \frac{\log(\Xzs(t))}{t}=
\lim_{t \to \infty} \frac{\log(\Xzi(t))}{t},\text{ a.s.},
    \end{split}
\end{equation}
for any $\cprfi{\bze(t)}\in\AA_{(f,h)}$.
\end{enumerate}
\item {\bf Relative constraints}\\\noindent \revision{Suppose that the function
  $h$ satisfies the assumption (B) from Definition
  \ref{def:port-corr}. }
Define the process $\cprfi{\bzr(t)}$ as a projection of the
Merton proportion $\bze_M(t)$ onto the (wealth-independent) constraint set
$F(t)$, under the metric $d_{\bsi}$.
Then $\bzr(t)$ is both $\log$- and growth-optimal in the class
$\AA_{(f,h)}$, i.e.,
\begin{equation}
    \nonumber 
    \begin{split}
\liminf_{t\to\infty} \frac{\log(\Xze(t))}{t}\leq
\liminf_{t\to\infty} \frac{\log(\Xzr(t))}{t}
, \text{ a.s.,}
\end{split}
\end{equation}
and
\begin{equation}
    \nonumber 
    \begin{split}
\EE[\log(\Xzr(\bt))-\log(\Xze(\bt))]&\leq 0, 
    \end{split}
\end{equation}
for all $\bze\in\AA_{(f,h)}$, and all $[0,\infty)$-valued 
stopping times $\bt$, interpreted as
time-horizons.
\end{enumerate}
\end{theorem}
\revision{
  \begin{remark} The central message of the main Theorem
    \ref{thm:main} is the following: even though the absolute
    constraints mix the wealth dependence and the risk-constraints in
    a complicated way, it turns out that the problem still admits a
    simple solution - just project the unconstrained optimal portfolio
    proportion onto the constraint set. Moreover, our analysis shows
    that in the conjunction with the ergodic criterion, the absolute
    wealth constraints are (eventually) so strong that the agent is
    forced to invest in a severely restricted way. In the case of a
    VaR-constraint, for example, no loss whatsoever is tolerated (in
    the asymptotic sense). On the other hand, we provide another
    optimal policy $\bze^*$ which performs much better on finite
    horizons, but attains the same asymptotic growth. Finally, we
    provide an explicit formula for the optimal asymptotic growth
    which depends in a simple way on the primitives of the model.
\\\indent
In the relative
    case, things are much simpler,  and we show that asymptotic
    optimality is equivalent to finite-horizon optimality for any
    choice of the horizon. The results obtained generalize directly
    the related results in \cite{CuoHeIss07}. 
\end{remark}
}
\revision{
\subsection{Some explicit examples}
Before we present the proof of Theorem \ref{thm:main} in the following
section, we illustrate some of its features through an example where the optimal
asymptotic growth-rates 
can be computed explicitly.
\begin{example}\ 
  \begin{enumerate}
  \item {\bf Constant coefficients.} \revision{Suppose that the coefficients $\bmu(t)\equiv
    \bmu$ and $\bsi(t)\equiv \bsi$ are constant. In that case the ergodic
    Assumption \ref{ass:ergodicity} is trivially satisfied, $Z(\vp)$ is
    a constant random variable for each $\vp$, and  and we have
    $Z(\vp)=\vp(\norm{\bze_M^T\bsi})$.  Therefore, 
\[\lim_{t \to \infty} \frac{\log(\Xzs(t))}{t}=
\lim_{t \to \infty} \frac{\log(\Xzi(t))}{t}=r+\norm{\bze_M^T\bsi}^2 \delta(\norm{\bze_M^T\bsi}).
\]
In the case where the constraints are such that $\bze_M\in F(t,x)$ for
all $t,x$, we clearly have $\delta(x)=1$, for all $x$ and we recover
the well known Merton's solution to the growth-rate
optimization problem. In the case of  
VaR-, TVar- and LEL-constraints,  the explicit expression for
$\delta$ (and, thus, for $\beta^*$ and $\beta^{\infty}$)
 can be obtained from the explicit expressions in Proposition
\ref{pro:formulas-for-risk-measures} and the representation 
(\ref{equ:equat-for-delta}) from Theorem \ref{thm:main}. While
elementary, these calculations are quite tedious and their results are
not very illuminating, so we omit them.}
  \item {\bf Periodic coefficients.} \revision{
  In this case, the coefficient processes $\bmu(t)$
and $\bsi(t)$ are assumed to be deterministic and periodic with period
$T_0$. It is not hard to see that the 
    Assumption \ref{ass:ergodicity} is still satisfied and that we
    have
$ Z(\vp)= \frac{1}{T_0} \int_0^{T_0} \vp( \norm{ \bze^T_M(t)
  \bsi(t)})\, dt,$ so that
\[\lim_{t \to \infty} \frac{\log(\Xzs(t))}{t}=
\lim_{t \to \infty} \frac{\log(\Xzi(t))}{t}=r+\frac{1}{T_0}
\int_0^{T_0} \norm{\bze_M^T(t)\bsi(t)}^2
\delta(\norm{\bze_M^T(t)\bsi(t)})\, dt.
\]}
\item {\bf Stochastic volatility.} \revision{While the calculations with the
  realistic constraints like VaR, TVar and LEL are possible, but quite
  messy in the stochastic volatility model as presented in Example
  \ref{exa:ergodicity} (2), the unconstrained case can be treated with
  ease. Indeed, then $\delta(x)=1$, for all $x$ and, using the
  discussion in Example \ref{exa:ergodicity} (2), we have
\[ \lim_{t \to \infty} \frac{\log(\Xzs(t))}{t}=
\lim_{t \to \infty} \frac{\log(\Xzi(t))}{t}=r+\frac{\mu^2}{\bar{\sigma}^2}, 
\]
where
\begin{equation}
    \nonumber 
    \begin{split}
\bar{\sigma}=  \left( \int_{\R} \frac{\sqrt{\nu}}{\sqrt{\pi}} 
 \Sigma(x)^{-2} e^{-\nu(x-\bar{V})^2} \, dx\right)^{-2}.
    \end{split}
\end{equation}
In words, the optimal growth-rate in the stochastic volatility market
matches the optimal growth rate in a constant-coefficient market in
which the volatility is a harmonic-type mean of the stochastic
volatility $\Sigma(\cdot)$ over the invariant measure.  
}
  \end{enumerate}
\end{example}
}
\section{Analysis}
\label{sec:analysis}

From this point onward, we fix a pair of functions $(f,h)$ as in
Definition \ref{def:port-corr}, and 
drop all related 
subscripts from the notation.  If a
distinction between the relative and the absolute case is needed, it
will be made explicit, and the unified notation $F(t,x)$ will be used
instead
of $F(t)$ for the relative constraints. 
 Unless stated otherwise, 
statements and definitions made for the values of random processes,
fields and correspondences are assumed to hold for all $t\in
[0,\infty)$, a.s.

\subsection{Properties of the constraint sets}\ 
Several analytical properties of the (instantaneous) constraint sets
$F(t,x)$ are established in this Subsection. 
\begin{lemma}
\label{lem:cbs}
For a vector $\bmu\in\R^n$ and a full-rank matrix $\bsi\in\R^{n\times
  m}$, let $\bze_M=(\bsi \bsi^T)^{-1} \bmu$. Then  the
following inequality holds for each $\bze\in\R^n$
\begin{equation}
    \nonumber 
    \begin{split}
\bzt\bmu\leq 
\norm{\bzt\bsi}\,\norm{\bze_M^T\bsi}.
    \end{split}
\end{equation} 
\end{lemma}
\begin{proof}
Since $\bmu=\bsi \bsi^T \bze_M$, we have
\begin{equation}
    \nonumber 
    \begin{split}
 \bzt\bmu= (\bsi^T \bze)^T (\bsi^T \bze_M)\leq
 \norm{\bzt \bsi}\,\norm{\bze_M^T \bsi},
    \end{split}
\end{equation}
by the Cauchy-Buniakowski-Schwarz inequality. 
\end{proof}

The following lemma gives an upper bound on the size of the constraint
sets. 
\begin{lemma}
\label{lem:C-bounds}
 There exist constants $C_i>0$, $i=1,2,3$ (independent of the market
 coefficients) such that 
\begin{equation}
    \nonumber 
    \begin{split}
\norm{\bzt \bsi(t)}\leq
 C_1\norm{\bzt_M(t)\bsi(t)}+C_2\sqrt{h(x)+C_3},
    \end{split}
\end{equation}
whenever $\bze\in F(t,x)$.  

Consequently, each $F(t,x)$ is contained
in a $d_{\bsi(t)}$-ball of (possibly infinite) radius
$C_1\norm{\bzt_M(t)\bsi(t)}+C_2\sqrt{h(x)+C_3}$ around the origin.
\end{lemma}
\begin{proof}
Without loss of generality we assume that $h(x)<\infty$.
Lemma \ref{lem:cbs} in conjunction with property \eqref{ite:kappas}
from
Definition \ref{def:port-corr}
yields that, for each $\bze\in F(t,x)$,  we have
\begin{equation}
    \nonumber 
    \begin{split}
0& \geq f(\bzt\bmu(t), \norm{\bzt\bsi(t)}\geq 
\kappa_1 \norm{\bzt\bsi(t)}^2-\kappa_2 \bzt\bmu(t)-\kappa_3-h(x)\\
 & \geq \kappa_1  \norm{\bzt\bsi(t)}^2- \kappa_2 \norm{\bzt_M(t)\bsi(t)}
 \norm{\bzt\bsi(t)}-\kappa_3-h(x),
    \end{split}
\end{equation}
for some constants $\kappa_i>0$, $i=1,2,3$. Consequently, a simple estimate
based on the quadratic inequality for $\norm{\bzt\bsi(t)}$ above and the fact
that $\sqrt{a+b}\leq \sqrt{a}+\sqrt{b}$ for $a,b\geq 0$  yields
\begin{equation}
    \nonumber 
    \begin{split}
      \norm{\bzt\bsi(t)}\leq C_1\norm{\bzt_M(t)\bsi}+C_2
      \sqrt{h(x)+C_3},\text{ where } k_1=\frac{\kappa_2}{\kappa_1},C_2
      =\sqrt{\frac{\kappa_3}{\kappa_1}},C_3=\kappa_3.
    \end{split}
\end{equation}
\end{proof}
The following Proposition is a simple 
 corollary of Lemma \ref{lem:C-bounds} above is 
\begin{proposition}
  In the relative case, the constraint set $F(t)$ is convex and
  compact. In the absolute case, $F(t,x)$ is always convex and either
  compact or equal to the whole $\R^n$, depending on whether $x>x_0$ or
  $x\leq x_0$.
\end{proposition}
\begin{proof}
  It is clear that $F=F(t,x)$ equals the whole of $\R^n$
  when $h(x)=+\infty$. We can suppose, therefore, that $h(x)\in\R$,
  treat both absolute and relative cases together, and establish
  compactness and convexity.

Convexity is
  inherited directly from the joint convexity of the function
  $(\zmu,\zsigma)\mapsto f(\zmu,\zsigma)$, its increase in the second
  variable, and the convexity of the
  mappings $\bze\mapsto \bzt \mu$ and $\bze\mapsto \norm{\bzt \bsi}$.
  
To establish compactness, we turn to Lemma \ref{lem:C-bounds} and
conclude that $F(t,x)$ is a bounded set, since the metric
$d_{\bsi(t)}$ and the Euclidean metric $d$ are equivalent. 
Finally, closedness of $F(t,x)$ follows from joint continuity of the
function $f$.
\end{proof}

\subsection{Structure of the projections on the constraint sets}
Proposition \ref{pro:projection} below
 exposes an interesting property of the
$d_{\bsi(t)}$-projection of the Merton-proportion process $\bze_M(t)$ onto
the constraint set $F(t,x)$ - namely, that it 
is collinear with $\bsy{0}$
and $\bze_M(t)$ and lies between them. 
This unexpected property of the constraints is going
to be instrumental for the arguments in the sequel. In preparation 
for the proof of Proposition \ref{pro:projection},
we need to introduce a random field $g:\Omega\times [0,T] \times [0,\infty)\to
\R$ and identify some of its properties; for $\beta\in [0,\infty)$ we
set
\begin{equation}%
\label{equ:function-g}
    \begin{split}
g(t,\beta)=  f(\beta \norm{\bzt_M(t)\bsi(t)}^2 , \beta
\norm{\bzt_M(t)\bsi(t)}),\ t\in [0,\infty).
    \end{split}
\end{equation}

\begin{lemma}
\label{lem:properties-function-g}
The following hold true for the random field $g$, defined in
\eqref{equ:function-g}.
\begin{equation}
    \nonumber 
    \begin{split}
(1) & \text{ for every $(\omega, t)\in \Omega\times [0,\infty)$, 
$g(t,\cdot)$ is a convex, continuously differentiable 
 function.}\\
(2) & \text{ 
for every $(\omega, t)\in \Omega\times [0,\infty)$,   
$g(t,0)=g(0,0)<0$, and}\\
(3) & \text{ for every $\beta>0$, $T>0$,
$\sup_{t\in[0,T]} \abs{g(t,\beta)}<\infty$, a.s.}\\
(4) & \text{ for every $(\omega, t)\in \Omega\times [0,\infty)$, and
  every $c>0$, the equation 
$g(t,\beta)=c$ has a unique solution.}
    \end{split}
\end{equation}
\end{lemma}
\begin{proof}
  Property (1) follows from the joint convexity of $f$, (2) is a
  restatement of the fact that $f(0,0)<0$, and (3) is a consequence of
  continuity of $f$, coupled with the local boundedness of the market
  coefficients.  To establish (4), we recall that $g(t,\cdot)$ is
  convex, $g(t,0)<0$ and $\lim_{\beta\to\infty} g(t,\beta)=+\infty$,
  thanks to the equation \eqref{equ:lower-bound-on-f} in Definition
  \ref{def:port-corr}.
\end{proof}

\begin{proposition}
\label{pro:projection}
Choose  $x\in (0,\infty]$, and 
let $\pi_F(\bze_M(t))$ denote the projection of the Merton-proportion
 $\bze_M(t)$ process onto 
 the convex set $F(t,x)$, with respect to the metric
 $d_{\bsi(t)}$. 
Then there exists a constant $\beta(t,x)$ - defining a
random field   $\beta:\Omega\times[0,T]\times
[0,\infty)\to (0,1]$  - such that
\begin{equation}
    \nonumber 
    \begin{split}
 \pi_F(\bze_M(t))=\beta(t,x) \bze_M(t).
    \end{split}
\end{equation}
Moreover, $\beta(t,x)=1$ when $h(x)=+\infty$. Otherwise, 
$\beta(t,x)=1\wedge b(t,x)$, where
$b(t,x)$ uniquely satisfies
\begin{equation}
    \nonumber 
    \begin{split}
g(t,b(t,x))= h(x).
    \end{split}
\end{equation} 
\end{proposition}
\begin{proof}
Existence and uniqueness of the projection $\pi_F(\bze_M(t))$
are consequences of the the compactness of the set $F(t,x)$ and strict
convexity of the norm $d_{\bsi(t)}$. 
All statements of the proposition are 
 trivial if  $\bze_M(t)\in F(t,x)$, so we can freely assume that
$\bze_M(t)\not\in F(t,x)$ . In particular, this assumption forces
$h(x)<\infty$.

The mapping $\bze\mapsto \bsi(t)^T\bze$ from $\R^n$ to
$\Range(\bsi(t)^T)\subseteq \R^m$ is a linear isomorphism.  Moreover, it is
also an isometry when $\R^n$ is equipped with the metric $d_{\bsi(t)}$
and $\Range(\bsi(t)^T)$ with the standard Euclidean metric $d$.  Therefore,
the image $\bsi(t)^T \pi_F(\bze_M(t))$ of the projection
$\pi_F(\bze_M(t))$ is the Euclidean projection
$\pi_{F'}(\tbmu)$ of the image $\tbmu=\bsi(t)^T \bze_M(t)$,
 onto the image $F'(t,x)=\bsi(t) F(t,x)$ 
of the constraint set $F(t,x)$. Consequently, it will be
enough to show that $\pi_{F'}(\tbmu)$ is of the form $\beta(t,x)
\tbmu$, and that the mapping $\beta$ has the desired properties.

With $\norm{\brho}_{\bsi(t)}$ defined as $d_{\bsi(t)}(\brho,\bsy{0})$,
for $\brho\in\Range(\bsi(t)^T)$,  we have 
\begin{equation}
    \nonumber 
    \begin{split}
 F(t,x) & =
\sets{\bze\in\R^n}{f(\bzt\bmu(t),\norm{\bzt\bsi(t)})\leq h(x)}, \text{
  and }\\
F'(t,x) &= \sets{\brho\in\Range(\bsi(t)^T)}{f(\brho^T \tbmu ,
\norm{\brho}_{\bsi(t)})\leq h(x)}.
    \end{split}
\end{equation}
The equality $\bze_M(t)\bmu(t)=\norm{\bze_M(t)\bsi(t)}^2$ and  Lemma
\ref{lem:properties-function-g} imply that there exists unique
number
$\beta(t,x)\in (0,1]$ such that 
 \begin{equation}%
\label{equ:various-betas}
    \begin{split}
\begin{cases} 
\beta\tbmu\in F'(t,x), & \text{ for } \beta \in [0,\beta(t,x)],
  \text{ and} \\
\beta\tbmu\not\in F'(t,x), & \text{ for } \beta\in (\beta(t,x),1].
\end{cases}
    \end{split}
\end{equation}
Moreover, remembering the assumption $\bze_M(t)\not\in F(t,x)$ (or,
equivalently $\tbmu\not\in F'(t,x)$), we can
easily see that
$\beta(t,x)<1$ must be of the form 
$\beta(t,x)=g(t,x)$ (where $g$ is defined in \eqref{equ:function-g}
above).

It is our goal to show that $\brho_0(t)\triangleq \beta(t,x)\tbmu$
coincides with the projection $\pi_{F'}(\tbmu)$. To progress with this
claim, let $P$ denote the semi-space 
\begin{equation}
    \nonumber 
    \begin{split}
P=\sets{\brho\in \Range(\bsi(t)^T)}{(\brho-\brho_0(t))^T \tbmu> 0},
    \end{split}
\end{equation}
supported by a hyperplane through $\brho_0(t)$, perpendicular to $\tbmu$.
Thanks to the assumption $\tbmu\not\in F'(t,x)$, the vector 
$\tbmu$ cannot be equal to
$\bsy{0}$, and so $P$ does not degenerate to the whole $\Range(\bsi(t)^T)$.
The points in $P^c\setminus\set{\brho_0(t)}$ 
are further away from $\tbmu$ than $\brho_0(t)$ is, so
it will be enough to show that $P\cap F'(t,x)=\emptyset$. 
Suppose, to the contrary, that there exists a vector 
$\tbrho\in P$ such that $\tbrho\in F'(t,x)$, i.e.,  
$f(\tbrho^T\tbmu,\norm{\tbrho}_{\bsi(t)})\leq h(x)$. Let 
\begin{equation}
    \nonumber 
    \begin{split}
\hro= \norm{\tbrho}_{\bsi(t)} \frac{\tbmu}{\norm{\tbmu}_{\bsi(t)}}.
    \end{split}
\end{equation} 
Since $\hro^T
\tbmu=\norm{\tbrho^T}_{\bsi(t)}\norm{\tbmu}_{\bsi(t)}\geq
\tbrho^T\tbmu$, and since the function $f$ is decreasing in its first
variable, we have
\begin{equation}
    \nonumber 
    \begin{split}
 f(\hro^T \tbmu, \norm{\hro^T} )= f(\hro^T \tbmu, \norm{\tbrho}) \leq 
f(\tbrho^T \tbmu, \norm{\tbrho})\leq h(x)
    \end{split}
\end{equation}
so $\hro\in F'(t,x)$. All three points $\bsy{0}$, $\hro$ and $\brho_0(t)$ 
are non-negative multiples of $\tbmu$, with $\brho_0(t)$ being between
the other two. 
Therefore, there exists a constant $\ld\in [0,1]$ such
that $\brho_0(t)=\ld \bsy{0}+(1-\ld) \hro$. Because $\hro\in P$,
$\ld>0$. Thanks to the joint convexity
of the function $f$, we have
\begin{equation}
    \nonumber 
    \begin{split}
f(\brho_0(t)^T\tbmu,\norm{\brho_0(t)})& =
f((1-\ld) \hro_0^T\tbmu,\norm{(1-\ld)\hro})\\
&=
f( \ld\bsy{0}+(1-\ld)\hro^T\tbmu,\ld
\norm{\bsy{0}}+(1-\ld)\norm{\hro})\\
&\leq \ld f(0,0)+ (1-\ld) f(\hro^T\tbmu,\norm{\hro})< h(x). 
    \end{split}
\end{equation}
Continuity of the mapping $\kappa\mapsto f(\kappa
\tbmu^T\tbmu,\norm{\kappa \tbmu_0})$ (from Lemma
\ref{lem:properties-function-g}) implies that there exists an
open interval $(\underline{\kappa},\overline{\kappa})$ around $\beta(t,x)$
such that $\kappa \tbmu\in F'$ for all $\kappa\in 
(\underline{\beta},\overline{\beta})$. This is, however, in
contradiction with \eqref{equ:various-betas}.
\end{proof}
\begin{remark}
The proof of Proposition \ref{pro:projection} above can be given
without a recourse to the change-of-variable transformation
$\brho=\bsi(t)^T \bze(t)$. We do this in order to help the reader's
intuition by placing him or her 
in the familiar isotropic Euclidean setting. 
\end{remark}

The following Lemma plays a central role in the proof of Lemma
\ref{lem:SDE-exists} below. It
establishes a uniform version of the Lipschitz property for the mapping
$\beta(\cdot,\exp(\cdot))$.
\begin{lemma}
\label{lem:Lipschitz}
There exists an increasing process $L:[0,\infty)\to (0,\infty)$ such that
\begin{equation}
\label{equ:beta-Lipschitz}
\abs{\beta(t,e^{y_1})-\beta(t,e^{y_2})} \leq L(T) |y_2-y_1|\ ,\forall\,
t\in[0,T],\  y_1,y_2 \in \R,
\end{equation}
for each $T>0$.
\end{lemma}
\begin{proof}
We fix a time horizon $T>0$, a time instance $t\in [0,T]$, and a
typical $\omega\in\Omega$.  Without loss of generality we may assume
that $x_1=\exp(y_1)<x_2=\exp(y_2)$ 
and that $\beta(t,x_2)<1$, which, in turn, implies that
$h(x_2)<\infty$. 
When $\beta(t,x_1)=1$ then $g(t,1)\leq h(x_1)$, and we can find a
unique $\bar{x}_1>x_0$ with the property that
$g(t,1)=h(\bar{x}_1)$. Clearly $x_2>\bar{x}_1\geq x_1$. When
$\beta(t,x_1)<1$, we simply set $\bar{x}_1=x_1$. In either case we
have
\begin{equation}%
\label{equ:difference-of-hs}
    \begin{split}
 h(\bar{x}_1)-h(x_2)=g(t,\beta(t,\bar{x}_1))-g(t,\beta(t,x_2))=
g(t,\beta(t,x_1))-g(t,\beta(t,x_2)).
    \end{split}
\end{equation}
Thanks to \eqref{equ:difference-of-hs} and 
continuous differentiability and convexity of the function
$g(t,\cdot)$ (see Lemma \ref{lem:properties-function-g}), we have
\begin{equation}%
\label{equ:estimate-on-g}
    \begin{split}
h(\bar{x}_1)-h(x_2)&=\int_{\beta(t,x_2)}^{\beta(t,x_1)} \pds{\beta}
g(t,\xi)\, d\xi\geq (\beta(t,x_2)-\beta(t,x_1))
\pds{\beta}g(t,\beta(t,x_2))\\
&\geq (\beta(t,x_2)-\beta(t,x_2))
 \frac{g(t,\beta(t,x_2))-g(0,0)}{\beta(t,x_2)}\\
&\geq
 (\beta(t,x_2)-\beta(t,x_1))
\frac{- g(0,0)}{\beta(t,x_2)},
    \end{split}
\end{equation}
where the last two inequalities follow from the convexity of $g(t,\cdot)$
and the fact that $g(t,0)=g(0,0)$, for any $t$.
On the other hand, due to the convexity of
$\tilde{h}(\cdot)=h(\exp(\cdot))$, we have
\begin{equation}%
\label{equ:estimate-on-h}
    \begin{split}
h(\bar{x}_1)-h(x_2)\leq (y_2-\log(\bar{x}_1)) (-\tilde{h}'(\log(\bar{x}_1) ))\leq
(y_2-y_1) (-\tilde{h}'(y_T)),
    \end{split}
\end{equation}
where $y_T=\log(h^{-1}(\sup_{t\in [0,T]} g(t,1)))>\log(x_0)$. Finally, as
$\beta(t,x_2)<1$, \eqref{equ:estimate-on-g} and
\eqref{equ:estimate-on-h}
can be combined to imply
\begin{equation}
    \nonumber 
    \begin{split}
\abs{\beta(t,e^{y_2})-\beta(t, e^{y_1})}=
\beta(t,e^{y_1})-\beta(t,e^{y_2})\leq L(T) (y_2-y_1)= L(T) \abs{y_2-y_1},
    \end{split}
\end{equation}
where $L(T)= \tilde{h}'(y_T)/ g(0,0)$. 
\end{proof}
\subsection{Candidate optimal portfolio proportions}
\begin{lemma}
\label{lem:SDE-exists}
The following stochastic differential equation
\begin{equation}%
\label{equ:the-SDE}
\left\{    \begin{split}
 d\Xzs(t)&=\Xzs(t)\big[\big(r+(\bzs)^T(t)\bmu(t)\big)\,dt
+(\bzs)^T(t)\bsi(t)\,d\bW(t)\big],\\
 & \qquad\qquad \text{where } \bzs(t) =\beta(t,\Xzs(t))\bze_M(t),\\
\Xzs(0)&=X(0)
    \end{split}\right.
\end{equation}
has a unique strong solution in $[0,\infty)$.
\end{lemma}
\begin{proof}
 It will be enough to choose a  fixed, but arbitrary time horizon
$[0,T]$, and 
prove existence and uniqueness of the solution  
$Y(t)=\log \Xzs(t)$ 
of the
stochastic differential equation 
\begin{equation}%
\label{equ:log-SDE}
   \left\{ \begin{split}
 dY(t)&=Q(t,\bzs(t))\,dt
+(\bzs)^T(t)\bsi(t)\,d\bW(t)\\
\bzs(t)&=
\beta(t,e^{Y(t)})\bze_M(t).
    \end{split}\right.
\end{equation} 
According to \cite[Theorem 7., p.~194]{Pro04} it will be enough to 
establish the Lipschitz property of the (c\' agl\' ad) coefficients of
\eqref{equ:log-SDE}, for each $\omega$, uniformly in $t\in [0,T]$.
 In that direction, we note that the coefficient $(\bzs)^T(t)\bsi(t)$
of $d\bW(t)$ satisfies the mentioned Lipschitz property thanks to 
Lemma \ref{lem:Lipschitz} and local boundedness of $\bsi(t)$.
 As for the $dt$-coefficient $Q(t,\bzs(t))$, we only need to 
observe that 
\begin{eqnarray}\label{equ:difference-of-Qs}\lefteqn{
\abs{Q(t,\beta(t,e^{y_2}))\bze_M(t)-Q(t,\beta(t,e^{y_1}))\bze_M(t)}
}\\\notag&=&\abs{\beta(t,e^{y^1})-\beta(t,e^{y_1})}\abs{1-\frac{1}{2}(\beta(t,e^{y^1})+\beta(t,e^{y_1}))}\abs{\bzt_M(t)
  \bmu(t)},
   \end{eqnarray}

and use Lemma \ref{lem:Lipschitz} and local boundedness of the process
$\abs{\bzt_M(t)\bmu(t)}$.
\end{proof}
We introduce the process $\cprfi{\bzi(t)}$, given by   
$\bzi(t)=\beta(t,\infty) \bzt_M(t)$ where 
$\beta(t,\infty)=\lim_{x\to\infty} \beta(t,x)$. It is readily seen
that $\bzi(t)$ is the $d_{\bsi(t)}$-projection of $\bze_M(t)$ onto the
limiting constraint set $F(t,\infty)$. Thanks to the previous Lemma, 
the process $\cprfi{\bzs(t)}$ is uniquely determined by
\eqref{equ:the-SDE}.
\begin{corollary}
$\cprfi{\bzs(t)}, \cprfi{\bzi(t)}\in \AA$.
\end{corollary}
\subsection{The question of transience}\ 
Before engaging in the proof of optimality of $\bzs(t)$ and
$\bzi(t)$, we need to understand better the transience properties of
the wealth process $\Xze(t)$ for arbitrary $\bze\in\AA$.
\begin{lemma}
\label{lem:llln-on-transient}
\label{van} 
For $\bze(t)\in\AA$, let  $\cprfi{\Xze(t)}$ be the corresponding wealth
process. Then
\begin{equation}
    \nonumber 
    \begin{split}
 \lim_{t\to\infty} \frac{\Mze(t)}{t}=0,
\text{ on $\set{\lim_{t\to\infty}
    \Xze(t)=\infty}\in\FF_{\infty}$, }
\end{split}
\end{equation}
where $\Mze(t)$ is the local martingale defined in
\eqref{equ:defined-Mze}.
\end{lemma}
\begin{proof}
Let $A(t)=[\Mze(t),\Mze(t)]$ be the
quadratic variation of $\Mze(t)$. 
By Lemma \ref{lem:C-bounds} and Definition \ref{def:port-corr}, 
there exists constants $D_i>0$, $i-1,2,3$ such that
\begin{equation}%
\label{equ:ineq-for-M-one}
    \begin{split}
 A(t)\leq D_1 t+D_2\int_0^t\norm{\bzt_M(u)\bsi(u)}^2\, du+D_3\int_0^t
 h(\Xze(u))\, du. 
    \end{split}
\end{equation}
Of course, the estimate above is only useful for $(t,\omega)$ where
$h(\Xze(t))<\infty$. Fortunately, for each $\omega\in
\Az=\set{\lim_{t\to\infty} \Xze(t)=\infty}\in\FF_{\infty}$ there
exists $T'(\omega)>0$ such that $h(\Xze(t))<1$ for all $t>T'(\omega)$.
This is a direct consequence of the definition of the set $\Az$ and the
properties of the function $h$.
Thus, the inequality \eqref{equ:ineq-for-M-one} can be transformed
into
\begin{equation}%
\label{equ:ineq-for-M-two}
    \begin{split}
 A(t)\leq A(T'(\omega))+D_1 (t-T'(\omega))+
D_2\int_{T'(\omega)}^t \norm{\bzt_M(u)\bsi(u)}^2\, du+D_3(t-T'(\omega))
    \end{split}
\end{equation}
for $t\geq T'(\omega)$ on $\Az$. 
 Assumption \ref{ass:ergodicity} now implies that
\begin{equation}
    \nonumber 
    \begin{split}
 \xi_0\triangleq \limsup_{t\to\infty} \frac{A(t)}{t} \leq D_1+D_2
 Z(x^2)+D_3<\infty,\text{ a.s.~on $\Az$,}
    \end{split}
\end{equation}
where the operator $Z$ is described in \eqref{equ:ergodic-limit}. 

By the Theorem of Dambis, Dubins and Schwarz (see
Theorem 4.6, p.~174 in \cite{KarShr91}), there exists a Brownian
motion $\cprfi{B_t}$ (possibly defined on the extended probability
space) such that $\Mze_t=B_{A(t)}$. By the Law of Large Numbers for
Brownian motion (see Problem 9.3, p.~104 in \cite{KarShr91})
 we have 
(with the convention $\tfrac{0}{0}=0$)\begin{equation}
    \nonumber 
    \begin{split}
\lim_{t\to\infty} \frac{\Mze(t)}{A(t)}=0,\text{ on $\Amz\triangleq
\set{\lim_{t\to\infty} A(t)=+\infty}$.}
    \end{split}
\end{equation}
On $(\Amz)^c=\set{\lim_{t\to\infty} A(t)<\infty}$, 
$\frac{\Mze(t)}{A(t)}$ converges to an a.s.-finite random variable $\xi_1$
(thanks to the continuity
property of the paths of the Brownian motion).
Finally,
\begin{equation}
    \nonumber 
    \begin{split}
 \limsup_{t\to\infty} \abs{\frac{\Mze(t)}{t}}\leq \limsup_{t\to\infty} \frac{\Mze(t)}{A(t)}
 \frac{A(t)}{t} =   
\left\{
\begin{array}{cl}
 0\cdot \xi_0, & \text{ on } \Amz\cap\Az \\
\xi_1\cdot 0, & \text{ on } (\Amz)^c\cap\Az
\end{array}
\right\}
=0\text{ on } \Az.
 \end{split}
\end{equation}
\end{proof}

Let $\cprfi{G(t)}$ be the $\R^n$-valued random correspondence 
defined by
\begin{equation}%
\label{equ:random-correspondence-G}
    \begin{split}
      G(t)\triangleq\sets{\bze\in\R^n}{ \bzt\bmu(t) \geq \tot \norm{\bzt\bsi(t)}^2}.
    \end{split}
\end{equation}
The set of all $\bze(t)\in\AA$ with the property that $ \bze(t)\in
G(t)$, for all $t\geq 0$, a.s., will be denoted by $\AA^G$.

\begin{lemma}
\label{lem:correspondence-G}
For each portfolio-proportion process $\cprfi{\bze(t)}\in\AA^G$, 
the wealth process $\cprfi{\Xze(t)}$ is transient,
i.e. $\lim_{t\to\infty} \Xze(t)=+\infty$, a.s.
\end{lemma}
\begin{proof}
Pick $\bze(t)\in\AA^G$, let $\Mze(t)$ be given by 
\eqref{equ:defined-Mze}, and $A(t)$ be as in the
proof of Lemma \ref{lem:llln-on-transient}. Following the mentioned
proof of Lemma \ref{lem:llln-on-transient}, we can write
$\Mze(t)=B_{A(t)}$ 
for some Brownian motion $\cprfi{B_t}$. 
Therefore, by It\^ o's lemma, 
\begin{equation}
    \nonumber 
    \begin{split}
\log(\Xze(t))=\log(X(0))+\int_0^t Q(u,\bze(u))\, du+
B_{A(t)}. 
    \end{split}
\end{equation}
The assumption $\bze(t)\in \AA^G$ implies that
\begin{equation}
    \nonumber 
    \begin{split}
 Q(t,\bze(t)) & = r+ \bzt(t)\bmu(t)-\tot
 \norm{\bzt(t)\bsi(t)}\geq r,\text{ for all $t>0$, a.s.,}
    \end{split}
\end{equation}
so the claim of the Lemma will follow once we establish the equality
\begin{equation}%
\label{equ:need-to-establish}
    \begin{split}
\lim_{t\to\infty} \frac{B_{A(t)}}{t}=0, \text{ a.s.}
    \end{split}
\end{equation}
By Lemma \ref{lem:cbs} combined with the assumption $\bze\in G(t)$ we
have
\begin{equation}
    \nonumber 
    \begin{split}
\tot \norm{\bzt(t)\bsi(t)}^2\leq \bzt(t)\bmu(t)\leq
\norm{\bzt_M(t)\bsi(t)} \norm{\bzt(t)\bsi(t)},
    \end{split}
\end{equation}
and so,
\begin{equation}%
\label{equ:bounded-above-by-zetaM}
    \nonumber 
    \begin{split}
A(t)=\int_0^t \norm{\bzt(u)\bsi(u)}^2\, du\leq 4 \int_0^t
\norm{\bzt_M(u)\bsi(u)}^2\, du,\text{ for all $t\geq 0$, a.s.}
    \end{split}
\end{equation}
By Assumption \ref{ass:ergodicity}, and the inequality 
\eqref{equ:bounded-above-by-zetaM} we see that 
\begin{equation}
    \nonumber 
    \begin{split}
\limsup_{t\to\infty} \frac{A(t)}{t}<\infty, \text{ for all $t\geq 0$, a.s.}
    \end{split}
\end{equation}
The remainder of the proof of the statement
\eqref{equ:need-to-establish} parallels the final argument of the
proof of Lemma \ref{lem:llln-on-transient}.
\end{proof}
\begin{lemma}
\label{lem:betas-coallesce}
Let $\cprfi{X(t)}$ non-negative process, and let
 $\Azx=\set{\lim_{t\to\infty} X(t)=+\infty}\in\FF_{\infty}$. Then 
\begin{equation}
    \nonumber 
    \begin{split}
\lim_{t\to\infty} (\beta(t,X(t))-\bi(t))=0,\text{ a.s. on $\Azx$.}
    \end{split}
\end{equation}  
\end{lemma}
\begin{proof}
Define the random process $\cprfi{\chi(t)}$ by
\begin{equation}
    \nonumber 
    \begin{split}
 \chi(t)=\begin{cases} 0, & \bze_M(t) \in F(t,\infty),\\
1, & \bze_M(t)\in F(t,X(t))\setminus F(t,\infty), \\
2 ,& \bze_M(t)\in F(t,X(t))^c
 \end{cases}
\end{split}
\end{equation}
so that $\bi(t)=1$ when $\chi(t)=0$, 
and $\bxt=1$ when $\chi(t)=0$ or $\chi(t)=1$. 
Thus, our task is reduced to the one of establishing the following
two claims
\begin{gather}%
\label{equ:claim-one}
\lim_{t\to\infty} (\bxt-\bi(t))\inds{\chi(t)=2}=0,\text{ a.s. on $\Azx$,
  and } \\
\label{equ:claim-two}
\lim_{t\to\infty} (1-\bi(t))\inds{\chi(t)=1}=0, \text{ a.s. on $\Azx$.}
\end{gather}

\noindent{\bf Claim \eqref{equ:claim-one}:} By the estimate
\eqref{equ:estimate-on-g} in Lemma \ref{equ:beta-Lipschitz},
for large enough $X(t)$ and $x>X(t)$ we have 
\begin{equation}
    \nonumber 
    \begin{split}
\inds{\chi(t)=2} \Big( \beta(t,X(t))-\beta(t,x)\Big) \leq 
\inds{\chi(t)=2}\frac{h(X(t))-h(x)}{-g(0,0)},
    \end{split}
\end{equation}
and, letting $x\to\infty$ yields,
\begin{equation}
    \nonumber 
    \begin{split}
0\leq \inds{\chi(t)=2} \Big( \beta(t,X(t))-\bi(t)\Big) \leq 
\inds{\chi(t)=2}\frac{h(X(t))}{-g(0,0)},
    \end{split}
\end{equation}
which, in turn, implies \eqref{equ:claim-one}.

\noindent{\bf Claim \eqref{equ:claim-two}:} 
 Let $A'$ be the set of 
all $\omega\in A$ for which the limit  in
\eqref{equ:claim-two} does not exist or 
differs from $0$. Fix a typical
 $\omega\in A'$, and
pick a sequence $\seq{t}$ such that $t_n\to\infty$ as $n\to\infty$, 
$\chi(t_n)=1$ for all $n\in\N$ and
$\lim_{n\to\infty} (1-\beta_n)\to l>0$, where $\beta_n=\bi(t_n)$. 
It is easily seen that $\kappa=\beta_n$ is the unique root of the equation 
\begin{equation}
    \nonumber 
    \begin{split}
f(\kappa \ld_n^2, \kappa \ld_n)=0, \text{ where }
\ld_n=\norm{\bze_M^T(t_n) \bsi(t_n)}.
    \end{split}
\end{equation}
Since $\chi(t)=1$, we know that $f(\ld_n^2, \ld_n)\leq
h(X(t_n))$. Thus, $\limsup_{n} f(\ld_n^2,\ld_n)\leq 0$.
 By joint convexity of $f$,
\begin{equation}
    \nonumber 
    \begin{split}
 0=f(\beta_n \ld_n^2, \beta_n \ld)\leq (1-\beta_n) f(0,0)+
\beta_n f(\ld_n,\ld_n^2).
    \end{split}
\end{equation}
Passing to the limit we get
\begin{equation}
    \nonumber 
    \begin{split}
0&\leq \limsup_{n} (1-\beta_n) f(0,0)+\beta_n f(\ld_n,\ld_n^2)
\leq \lim_n (1-\beta_n) f(0,0)=l f(0,0). 
    \end{split}
\end{equation}
This is in contradiction with the fact that 
$f(0,0)<0$, and we can conclude that
there
is no typical $\omega\in A'$.
\end{proof}
\subsection{Proving optimality}
We are finally ready to show that both $\bzs(t)$ and $\bzi(t)$ are
optimal. The first step is to identify the (common) 
value of those strategies. After that we show that no other strategy
can produce a higher value.
\begin{lemma} 
\label{lem:function-delta}
There exists a (deterministic)
function $\delta:[0,\infty)\to [0,1]$  such that 
\begin{equation}
    \nonumber 
    \begin{split}
\bs(t)=\delta(\norm{\bzt_M(t)\bsi(t)}).
    \end{split}
\end{equation}
\end{lemma}
\begin{proof}
It is a simple consequence of the regularity properties of the
functions $f$ and $h$ that $\bs(t)$ can be characterized as
$\bs(t)=\min(1,g^*(t))$, where $g^*(t)=\kappa(\norm{\bzt_M(t)
  \bsi(t)})$ 
is the unique solution of
\begin{equation}
    \nonumber 
    \begin{split}
 g(t,\kappa)=f( \kappa \norm{\bzt_M(t) \bsi(t)}^2 , \kappa \norm{\bzt_M(t) \bsi(t)} )=0.
    \end{split}
\end{equation}
Therefore, $\delta(\ld)=\min(1,
\kappa(\ld))$ is the sought-for function. 
\end{proof}
\begin{lemma}
\label{lem:limits-are-equa}
\begin{equation}%
\label{equ:limits-of-two-strategies}
    \begin{split}
\lim_{t\to\infty}\frac{\log(\Xzs(t))}{t}=
\lim_{t\to\infty}\frac{\log(\Xzi(t))}{t}=r+Z(x^2\delta(x)),  
    \end{split}
\end{equation}
where $Z(\cdot)$ is defined in Assumption \ref{ass:ergodicity}, and
$\delta$ is the function from Lemma \ref{lem:function-delta}.  
\end{lemma}

\begin{proof}
We first show that the limits in \eqref{equ:limits-of-two-strategies}
are equal. 
It is a matter of a simple calculation to show
 that both strategies $\cprfi{\bzs(t)}$ and
$\cprfi{\bzi(t)}$, belong to $\AA^G$, where $\AA^G$ is introduced after
\eqref{equ:random-correspondence-G}.  
Therefore, 
by Lemmas \ref{lem:llln-on-transient} and \ref{lem:correspondence-G},
\begin{equation}
    \nonumber 
    \begin{split}
 \lim_{t\to\infty} \frac{1}{t} \left(
   \log{\Xzs(t)}-\log(\Xzi(t))-
\int_0^t \Big( Q(u, \bzs(u))- Q(u,\bzi(u)) \Big)\, du \right)=0,\text{
a.s.,}
    \end{split}
\end{equation}
so it is enough to show that
\begin{equation}
    \nonumber 
    \begin{split}
      \lim_{t\to\infty} \frac{1}{t} \int_0^t \left[ Q(u, \bzs(u))-
        Q(u, \bzi(u))\right]\, du=0,\text{ a.s.}
    \end{split}
\end{equation}
Direct computation yields that 
\begin{equation}
    \nonumber 
    \begin{split}
 Q(t,\bzs(t))-Q(t,\bzi(t))=(\bs(t)-\bi(t))\left(1-\frac{1}{2}(\bs(t)+\bi(t))\right)\bze_M^T(t)\bmu(t), 
    \end{split}
\end{equation}
so thanks to the ergodic property of the process
$\cprfi{\bze_M^T(t)\bmu(t)}$ 
(Assumption \ref{ass:ergodicity}),
it will be enough to show that $\lim_{t\to\infty} (\bs(t)-\bi(t))=0$,
a.s. This, however, follows from Lemma
\ref{lem:betas-coallesce}.

To identify the limit, we use the Lemma \ref{lem:llln-on-transient} to
conclude that 
\begin{equation}
    \nonumber 
    \begin{split}
\lim_{t\to\infty} \frac{\log(\Xzs(t))}{t}&=r+\lim_{t\to\infty}
\frac{1}{t} \int_0^t \bs(u) \bze_M^T(u)\bmu(u)\, du\\ 
&=
r+\lim_{t\to\infty}\frac{1}{t} \int_0^t \delta(\norm{\bzt_M(u)\bsi(u)})
\bze_M^T(u)\bmu(u)\, du\\
&=r+\lim_{t\to\infty}\frac{1}{t} \int_0^t \delta(\norm{\bzt_M(u)\bsi(u)})
\norm{\bzt_M(u)\, \bsi(u)}^2 \, du=r+Z(x^2\delta(x)).
    \end{split}
\end{equation}
\end{proof}

\begin{lemma}
\label{lem:only-where-transient}
  For each $\cprfi{\bze(t)}\in\AA$ we have
\begin{equation}
    \nonumber 
    \begin{split}
 \liminf_{t\to\infty} \frac{\log(\Xze(t))}{t}= \begin{cases}
0, &\text{ on } \set{\liminf_{t\to\infty}\Xze(t)<\infty} \\
\liminf_{t\to\infty} \frac{1}{t} \int_0^t \tilde{Q}(t,\bze(t))\, dt, &
\text{ on } \set{\lim_{t\to\infty}\Xze(t)=\infty}
\end{cases}
    \end{split}
\end{equation}
\end{lemma}
\begin{proof}
  It\^ o's formula applied to the process $\cprfi{\log(\Xze(t))}$
  yields
\begin{equation}
    \nonumber 
    \begin{split}
 \frac{1}{t} \log(\Xze(t))=\frac{X(0)}{t}+\frac{1}{t} \int_0^t
 \tilde{Q}(u,\bze(u))\, du + \frac{1}{t} \int_0^t \bzt(u)\bsi(u)\, d\bW(u).
    \end{split}
\end{equation}
it remains to to let $t\to\infty$  and apply the result of Lemma
\ref{lem:llln-on-transient}. 
\end{proof}

\begin{theorem}\label{main12}
The portfolio-proportion process  $\bzi(t)$ is optimal, i.e., 
\begin{equation}
    \nonumber 
    \begin{split}
\liminf_{t\to\infty} \frac{1}{t} \log(\Xze(t))\leq \lim_{t\to\infty}
\frac{1}{t} \log(\Xzi(t))=r+Z(x^2\delta(x^2)),\text{ a.s., for each $\bze(t)\in\AA$,}
    \end{split}
\end{equation} 
where $\delta$ is the function introduced in Lemma
\ref{lem:function-delta}.
\end{theorem}
\begin{proof}
Pick $\bze(t)\in\AA$ and recall that, by Lemma
\ref{lem:only-where-transient} and strict positivity of the
parameter $r$, it will be enough to show that 
\begin{equation}
    \nonumber 
    \begin{split}
\liminf_{t\to\infty} \frac{1}{t} \int_0^t Q(u,\bze(u))\, du\leq
\liminf_{t\to\infty} \frac{1}{t} \int_0^t Q(u,\bzi(u))\, du,
\text{ on } \Az \triangleq \set{\lim_{t\to\infty} \Xze(t)=+\infty}. 
    \end{split}
\end{equation} 
Let $d_{\bsi(t)}$ is the metric on $\R^n$ defined in
\eqref{equ:dsig-defined} so that 
\begin{equation}%
\label{equ:Q-in-terms-of-d}
    \begin{split}
 Q(t,\bze)& =r+ \bzt \bmu-\tot\norm{\bzt\bsi}=
(r+\tot\norm{\bze_M^T\bsi}^2)-\tot d^2_{\bsi}(\bze,\bze_M).  
    \end{split}
\end{equation}
Furthermore, we have the following simple expression
\begin{equation}%
\label{equ:difference-of-Qs-two}
    \begin{split}
      Q(t,\bzi(t))-Q(t,\bze(t))= \tot \Big(
      d^2_{\bsi(t)}(\bze(t),\bze_M(t))-
      d^2_{\bsi(t)}(\bzi(t),\bze_M(t)) \Big).
    \end{split}
\end{equation}
Consequently, all we need to show is the following inequality 
\begin{equation}%
\label{equ:limsup}
    \begin{split}
 \limsup_{t\to\infty} \frac{1}{t} \int_0^t \big[ 
d^2_{\bsi(t)}(\bze(u),\bze_M(t))-d^2_{\bsi(t)}(\bzi(u),\bze_M(u))\big]\, du \geq 0, 
\text{ a.s.~on $\Az$}.
    \end{split}
\end{equation}
Being an element of $F(t,\Xze(t))$, the vector $\bze(t)$ is
$d_{\bsi(t)}$-further away from $\bze_M(t)$ than the projection
  $\beta(t,\Xze(t))\bze_M(t)$ of $\bze_M(t)$ onto $F(t,\Xze(t))$.
  Therefore, the expression inside the $\limsup$ in \eqref{equ:limsup}
  dominates the difference 
  $d^2_{\bsi(t)}(\beta(t,\Xze(t))\bze_M(u),\bze_M(t))-d^2_{\bsi(t)}(\bzi(u),\bze_M(u))$ of
  squared distances.
  Furthermore, this difference can be rewritten as
  \[(\beta(t,\Xze(t))-\bs(t))^2 \norm{\bzt_M(t)\bsi(t)}.\] It remains
  to employ the ergodicity Assumption \ref{ass:ergodicity}, and use
  the result of Lemma \ref{lem:betas-coallesce}, which states that
  $\beta(t,\Xze(t))-\bi(t)\to 0$ on $\Az$.
\end{proof}

\subsection{The relative constraints} 
We deal with the relative constraints in this last subsection. The
infinite-horizon ergodic optimization problem can be treated in a fashion virtually 
identical to the case of absolute constraints, so we leave it to the
interested reader. 

It remains to deal with the finite-horizon problem of optimal expected
logarithmic utility. As before, $\AA$ will denote a generic
admissibility set corresponding to a pair $(f,h)$ of functions
satisfying the assumptions in Definition \ref{def:port-corr} (with the
variant (A) for the function $h$).
 Moreover, we pick a time
horizon $T>0$. 

Define the process $\cprfi{\bze^r(t)}$ as a $d_{\bsi(t)}$-projection
of $\bze_M(t)$ onto the instantaneous constraint set $F(t)$. 
\begin{lemma}
\label{lem:inverse-triangle}
  For any $\bze\in F(t)$, the following inequality holds
\begin{equation}%
\label{equ:inverse-triangle}
    \begin{split}
d^2_{\bsi(t)}(\bze_M(t),\bze )
\geq d^2_{\bsi(t)}(\bze_M(t),\bzr(t))
+d^2_{\bsi(t)}(\bzr(t),\bze ).
    \end{split}
\end{equation}
\end{lemma}
\begin{proof}
  $\bze^r(t)$ is defined as the minimizer (in the convex set $F(t)$)
  of the distance $\gamma(\cdot)=d_{\bsi(t)}(\bze_M(t),\cdot)$.
  Therefore the directional derivative of the square $\gamma^2(\cdot)$
  in the direction $\bze-\bze^r(t)$, evaluated at the point
  $\bze^r(t)$, must be non-positive, i.e.,
\begin{equation}%
\label{equ:partial}
    \begin{split}
      0&\geq \nabla \gamma^2(\bze^r(t)) (\bze-\bze^r(t))=
      (\bze^r(t)-\bze_M(t))^T \bsi(t)
      \bsi(t)^T (\bze^r(t)-\bze)\\ &=\tot \Big( 
d^2_{\bsi(t)}(\bze_M(t),\bzr(t))+d^2_{\bsi(t)}(\bzr(t),\bze
      )-d^2_{\bsi(t)}(\bze_M(t),\bze )\Big).
    \end{split}
\end{equation}
\end{proof}

\begin{lemma}
  The process $\bze^r(t)$ is in $\AA$ and the quotient
\begin{equation}%
\label{lem:Y-is-supermartingale}
    \begin{split}
Y^r(t)=\frac{\Xze(t)}{\Xzr(t)},\ t\in [0,\infty)
    \end{split}
\end{equation}
is a strictly positive supermartingale for each $\bze(t)\in\AA$.
\end{lemma}
\begin{proof}
That $\bze^r(t)\in\AA$ follows directly from its construction as a
projection onto the instantaneous constraint set $F(t)$. 
Consequently, both
$\Xze(t)$ and $\Xzr(t)$ are strictly positive processes, and
therefore, so is $Y^r$. In order to show
that $Y^r$ is a supermartingale, we use the It\^ o's lemma and
expression \eqref{equ:Q-in-terms-of-d} to
conclude that its semimartingale decomposition of $Y^r(t)$ is of the
form 
\begin{equation}
    \nonumber 
    \begin{split}
dY^r(t)= -\tot \Big(
d^2_{\bsi(t)}(\bze_M(t),\bze(t))
-d^2_{\bsi(t)}(\bze_M(t),\bzr(t))
-d^2_{\bsi(t)}(\bzr(t),\bze(t)) \Big)\, dt+d L_t,
    \end{split}
\end{equation} 
where $L(t)$ is a local martingale. By Lemma
\ref{lem:inverse-triangle}, $Y^r(t)$ is a local supermartingale, and
its non-negativity allows us to use the standard argument based on the
Fatou Lemma to conclude that it is a (true) supermartingale.
\end{proof}

\begin{lemma}
Let $\bt$ be an $[0,\infty)$-valued stopping time. Then  
\begin{equation}
    \nonumber 
    \begin{split}
 \EE[\log(\Xzr(\bt)]\leq \EE[\log(\Xze(\bt))],
    \end{split}
\end{equation}
for any $\bze(t)\in\AA$.
\end{lemma}
\begin{proof}
By concavity of the function $\log(\cdot)$, we have
\begin{equation}
    \nonumber 
    \begin{split}
\EE[\log(\Xzr(\bt))-\log(\Xze(\bt))]\geq \EE[ (\Xzr(\bt)-\Xze(\bt))
\frac{1}{\Xzr(\bt)}]=1-\EE[ \frac{\Xze(\bt)}{\Xzr(\bt)}]\leq 0,   
    \end{split}
\end{equation}
where the last inequality follows from Lemma \ref{lem:Y-is-supermartingale} and the optional sampling
theorem. 
\end{proof}

\appendix
\section{Some technical results}
\label{sec:technical-results}
\begin{proof}[Proof of Proposition \ref{pro:formulas-for-risk-measures}]
The expression \eqref{equ:var} for $\var$ follows directly from 
its definition. In the case of $\tvar$, the conditional expectation
in \eqref{equ:tvar-def} can be written as
\begin{equation}
\tvar(x,\zmu,\zsigma)=
\left(\frac{x}{\alpha\sqrt{2\pi}} \int_{-\infty}^{N^{-1}(\alpha)}
\left[1-\exp\left\{Q(\zmu,\zsigma)\tau+y\zsigma
\sqrt{\tau}\right\}\right]e^{-\frac{y^2}{2}}\,dy\right)^+.
\end{equation}
This integral readily evaluates to \eqref{equ:tvar}. Finally, the
calculation of $\lel$ is identical to the one for $\tvar$ with
$\bmu=0$.
\end{proof}
\begin{proof}[Compliance with Definition \ref{def:port-corr}]
We only concentrate on the absolute case, as the relative one is
completely analogous and easier. For the

\medskip

\noindent{\bf  $\var$-constraint:} Take
\begin{equation}
    \nonumber 
    \begin{split}
\efa_V(\zmu,\zsigma)=-\tau(r+\zmu-\tot \zsigma^2)-N^{-1}(\alpha)\zsigma
\sqrt{\tau},\quad \eha_V(x)= -\log\left[ (1-\frac{a^{\mathrm{abs}}_V}{x})^+\right].
    \end{split}
\end{equation}
All of the properties (1)-(5) from the statement of the Proposition
can be obtained easily. 

\medskip

\noindent{\bf  $\tvar$-constraint:} Set
\begin{equation}%
\label{equ:tvar-f}
    \begin{split}
\efa_T(\zmu,\zsigma)=\log(\alpha)-\tau(r+\zmu)-\log(N(N^{-1}(\alpha)-\zsigma
\sqrt{\tau})),\quad \eha_T(x)= -\log(1-\frac{a^{\mathrm{abs}}_T}{x})^+.
    \end{split}
\end{equation}
For the $\tvar$ case, we only discuss the estimate
\eqref{equ:lower-bound-on-f}, while the other properties
follow simply from \eqref{equ:tvar-f}. For 
\eqref{equ:lower-bound-on-f} we simply note that 
\begin{equation}
    \nonumber 
    \begin{split}
\lim_{\zsigma\to\infty}
 \frac{ \log(N(N^{-1}(\alpha)-\zsigma \sqrt{\tau}))}{\zsigma^2}=-\tot\tau.
    \end{split}
\end{equation}

\noindent{\bf  $\lel$-constraint:} $\lel$ is a special case of $\tvar$ with
$\zmu=0$. 

\end{proof}

\def\cprime{$'$} \def\cprime{$'$}
\providecommand{\bysame}{\leavevmode\hbox to3em{\hrulefill}\thinspace}
\providecommand{\MR}{\relax\ifhmode\unskip\space\fi MR }
\providecommand{\MRhref}[2]{%
  \href{http://www.ams.org/mathscinet-getitem?mr=#1}{#2}
}
\providecommand{\href}[2]{#2}

\end{document}